\documentclass[preprint,12pt,nopreprintline,authoryear]{elsarticle}

\usepackage[colorlinks,citecolor=blue,urlcolor=blue]{hyperref}
\usepackage{amsthm,amsmath,amsfonts,amssymb,bbm}
\usepackage{accents}
\usepackage{algorithm}
\usepackage{algpseudocode}
\usepackage{graphicx}
\usepackage{mathtools}
\usepackage{subcaption}
\usepackage{tikz}
\usepackage{tabularray}
\usepackage{booktabs}
\usepackage{multirow}
\usepackage{hhline}
\usepackage{makecell}

\DeclarePairedDelimiter\floor{\lfloor}{\rfloor}

\DeclareMathOperator*{\argmin}{arg\,min}

\let\originalleft\left
\let\originalright\right
\renewcommand{\left}{\mathopen{}\mathclose\bgroup\originalleft}
\renewcommand{\right}{\aftergroup\egroup\originalright}

\newcommand{\norm}[1]{\left\lVert#1\right\rVert}

\newcommand{\expect}[2]{\mathbb{E}_{#1}\left[#2\right]}
\newcommand{\expec}[1]{\mathbb{E}\left[#1\right]}

\newcommand{\Var}[1]{\textrm{Var}\left[#1\right]}

\newcommand{\Cov}[1]{\textrm{Cov}\left[#1\right]}

\renewcommand{\P}[1]{\mathbb{P}\left[#1\right]}
\renewcommand{\d}{\text{d}}

\newcommand{\R}{\mathbb{R}}

\newcommand{\N}{\mathbb{N}}

\newcommand{\W}{{\mathcal{W}}}

\newcommand{\abs}[1]{\left|#1\right|}

\newcommand{\eset}[1]{{\left\{#1\right\}}}

\renewcommand{\t}[1]{{\textrm{#1}}}
\renewcommand{\c}[1]{{\mathcal{#1}}}

\newcommand{\Xts}{\eset{X_t}_{t\in\N}}

\newtheoremstyle{def}
  {12pt}   
  {6pt}   
  {\normalfont}  
  {0pt}       
  {\bfseries} 
  {.}         
  {5pt plus 1pt minus 1pt} 
  {}          

\theoremstyle{plain}
\newtheorem{theorem}{Theorem}[section]
\newtheorem{lemma}{Lemma}[section]

\newtheorem{example}{Example}[section]
\newtheorem{assumption}{Assumption}[section]

\newtheorem{proposition}{Proposition}[section]

\newtheorem{bproposition}{Proposition}[section]
\newtheorem{blemma}{Lemma}[section]

\newtheorem{cproposition}{Proposition}[section]
\newtheorem{clemma}{Lemma}[section]

\graphicspath{{figures/}}

\def\centerarc[#1](#2)(#3:#4:#5)
    { \draw[#1] ($(#2)+({#5*cos(#3)},{#5*sin(#3)})$) arc (#3:#4:#5); }

\definecolor{axis_color}{rgb}{.73,.73,.73}

\begin{document}

\begin{frontmatter}

\title{An Autoregressive Model for Time Series of Random Objects\tnoteref{title}}
\author{Matthieu Bult\'e\fnref{label1,label2}}
\ead{mb@math.ku.dk}
\author{Helle Sørensen\fnref{label1}}
\ead{helle@math.ku.dk}
\affiliation[label1]{organization={Department of Mathematical Sciences, University of Copenhagen}}
\affiliation[label2]{organization={Faculty of Business Administration and Economics, Bielefeld University}}

\begin{abstract}
  Random variables in metric spaces indexed by time and observed at equally spaced time points are receiving increased attention due to their broad applicability. The absence of inherent structure in metric spaces has resulted in a literature that is predominantly non-parametric and model-free. To address this gap in models for time series of random objects, we introduce an adaptation of the classical linear autoregressive model tailored for data lying in a Hadamard space. The parameters of interest in this model are the Fr\'echet mean and a concentration parameter, both of which we prove can be consistently estimated from data. Additionally, we propose a test statistic for the hypothesis of absence of serial correlation and establish its asymptotic normality. Finally, we use a permutation-based procedure to obtain critical values for the test statistic under the null hypothesis. Theoretical results of our method, including the convergence of the estimators as well as the size and power of the test, are illustrated through simulations, and the utility of the model is demonstrated by an analysis of a time series of consumer inflation expectations.
\end{abstract}


\end{frontmatter}
\section{Introduction}

Random variables in general metric spaces, also called random objects, have been receiving increasing attention in recent statistical research. The generality of the metric space setup does not require any algebraic structure to exist and is only based on the definition of a distance function. This allows the methods developed to be applied in domains ranging from classical setups to more complex use cases on non-standard data. This includes the study of functional data \citep{ramsay_functional_2005}, data lying on Riemannian manifolds, correlation matrices and applications thereof to fMRI data \citep{petersen_frechet_2019} or adjacency matrices and social networks \citep{dubey_frechet_2020} among others.

One example of particular interest due to its wide range of applications is that of data comprising of probability density functions. Probability distributions are a challenging example of a space that is both functional, and thus infinite-dimensional, but also non-Euclidean because of the constraints characterizing density functions. This leads to a number of different approaches to studying these objects: they have been studied as the image of Hilbert spaces under transformations \citep{petersen_functional_2016}, as specific Hilbert spaces with specific addition and scalar multiplication operators \citep{van_den_boogaart_bayes_2014}, as well as  metric spaces with stylized distances constructed to expose certain properties or invariances \citep{panaretos_invitation_2020, srivastava_functional_2016}. See \citet{petersen_modeling_2022} for a review of such methodologies. Distributions can be found in many applications, for instance the distribution of socioeconomic factors within a population such as income \citep{yoshiyuki_functional_2017}, fertility \citep{mazzuco_fitting_2015} or mortality data \citep{chen_wasserstein_2021}. They are also useful when considering belief distributions of economic factors \citep{meeks_heterogeneous_2023}, allowing economic analyses to consider entire distributions rather than empirical expectations. 

The study of random objects has received recent attention with work on hypothesis testing and inference \citep{dubey_frechet_2019, dubey_frechet_2020, mccormack_equivariant_2023, mccormack_stein_2022, kostenberger_robust_2023} as well as various approaches to regression \citep{petersen_frechet_2019, bulte_2024107995, hanneke_universal_2021}. Since the setup of general metric spaces offers little structure, part of the literature considers additional assumptions on the space in order for standard statistical quantities to be well-defined. This is often done by assuming that the metric space is a Hadamard space, see for instance \citet{auscher_probability_2003} for a detailed review of results in Hadamard spaces and \citet{bacak_computing_2014} for computation of Fr\'echet means in such spaces.

In many of the applications mentioned above, the data might be naturally observed repeatedly on a regular time grid, forming a time series. In this case, the observations might not be independent and the models and analyses require additional care to take this dependency into account. The existing literature in this setup has mainly been carried out in a non-parametric setting, with classical weak dependence assumptions. This has been done for instance for testing serial dependence \citep{jiang_testing_2023} or for proving the consistency of the Fr\'echet mean estimator \citep{caner_m-estimators_2006}. While this line of work can be broadly applied, they rely on non-parametric assumptions rather than proposing specific models.

On the other hand, time series models have been developed for specific non-Euclidean random objects by exploiting the structure of the space under study. One popular class of models is that of autoregressive models, which have been defined using the linear structure of functional spaces \citep{bosq_linear_2000, caponera_asymptotics_2021} or exploiting a tangent space structure of the space \citep{zhu_spherical_2022, xavier_generalization_2006,Ghodrati_jtsa12736, zhu_autoregressive_2021, jiang_wasserstein_2023} to name a few.

Inspired by existing autoregressive models, we propose an autoregressive model for random objects. Relying on an interpretation of iteration in the linear autoregressive model as a noisy weighted sum to the mean, we define a model parametrized by a mean and a concentration parameter. For this to be possible, we assume additional structure and require the space to be a Hadamard space, and exploit the geometry of the space to define the time series iteration through geodesics. We develop the methodology and associated theory for estimation and hypothesis testing in this model. This includes estimators for the mean and the concentration parameter, and we propose a test statistic for testing for no autocorrelation, corresponding to observing an i.i.d.\,sample, of which we characterize the asymptotic behavior under the null hypothesis and the alternative of a non-zero concentration parameter.

The paper is organized as follows: Section \ref{sec-prelim} gives a presentation of useful concepts and results in Hadamard spaces for the rest of the article. In Section \ref{sec-model}, we present our autoregressive model and present a theorem providing a sufficient condition for the existence of a stationary solution of the iterated system of equations associated with the model, and prove the identifiability of the model parameters. We propose in Section \ref{sec-estimation} estimators for these parameters and prove convergence results for those estimators. In Section \ref{sec-hypothesis-test}, we present our hypothesis test of independence. Finally, we illustrate our theoretical results with a numerical study in Section \ref{sec-numerical} and an application to real data in Section \ref{sec-application}.


\section{Preliminaries}\label{sec-prelim}

Let $(\Omega, d)$ be a metric space and $X$ a random variable, a Borel measurable function from some probability space to $\Omega$. We say that $X \in L^p(\Omega)$ if $\expec{d(X, \omega)^p} < \infty$ for some (and hence by the triangle inequality all) $\omega \in \Omega$. In the study of random objects, the concepts of mean and variance are generalized following the ideas of \citet{frechet_elements_1948}. Given a random variable $X \in L^2(\Omega)$, the Fr\'echet mean and variance of $X$ are defined as
\begin{equation}\label{eq-frechet-def}
    \expec{X} = \argmin_{\omega \in \Omega} \expec{d(X, \omega)^2}\qquad\Var{X} = \inf_{\omega \in \Omega} \expec{d(X, \omega)^2}.
\end{equation}
While the existence of the variance in Euclidean spaces implies the existence of a mean, this is not necessarily the case in general metric spaces. Furthermore, on its own, a metric space offers very little to define parametric models. We now present the additional structure that will be used in this work to construct models for time series of random objects following the presentation in \citet{burago_course_2001} and \citet{auscher_probability_2003}.

We call a map $\gamma : [0,1] \rightarrow \Omega$ a \textit{path} if it continuously maps the unit interval to $\Omega$. A path $\gamma$ such that $d(\gamma(r), \gamma(t)) = d(\gamma(r), \gamma(s)) + d(\gamma(s), \gamma(t))$ for every $r < s < t \in [0, 1]$ is called a \textit{geodesic}.
Given two elements $\omega, \omega' \in \Omega$, a path $\gamma$ is said to \textit{connect} $\omega$ and $\omega'$ if $\gamma(0) = \omega$ and $\gamma(1) = \omega'$. The set of all such paths is denoted by $\Gamma(\omega, \omega')$. The distance function $d$ induces a \textit{length} on the set of paths, defined for each $\gamma$ by
\begin{equation*}
    L_p(\gamma) = \sup \eset{ \sum_{i=1}^k d(\gamma(t_{i-1}), \gamma(t_i)) \mid 0 = t_0 \leq \ldots \leq t_k = 1, k \geq 1 }.
\end{equation*}
By definition, $d(\gamma(0), \gamma(1)) \leq L_p(\gamma)$, and hence $d(\omega, \omega') \leq \inf_{\gamma \in \Gamma(\omega, \omega')} L_p(\gamma)$ for every $\omega, \omega' \in \Omega$.
A metric space in which the previous inequality always holds as an equality is called a \textit{length space}. Furthermore, if there exists a geodesic $\gamma$ connecting each pair $\omega, \omega' \in \Omega$, then we see that the infimum is attained by $\gamma$, and $\Omega$ is called a \textit{geodesic space}.

\begin{figure}[t!]
    \centering
    \includegraphics[width=0.5\textwidth]{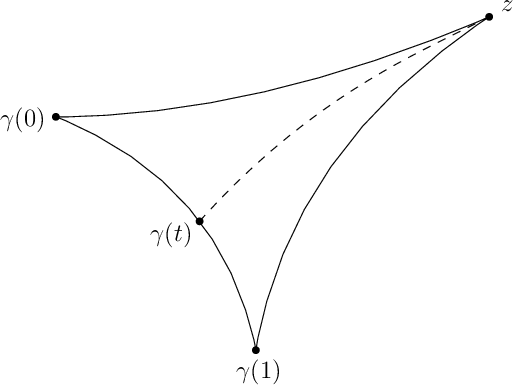}
    \caption{Illustration of the NPC inequality: the distance from any point on a side of the triangle to the opposite vertex is shorter than it would be in an equivalent Euclidean triangle.    }
    \label{fig-npc-triangle}
\end{figure}

A class of metric spaces of special interest are \textit{Hadamard spaces}. A metric space $(\Omega, d)$ is called Hadamard space if it is complete and satisfies the \textit{Non-Positive Curvature (NPC) inequality}: for each pair $\omega_0, \omega_1 \in \Omega$, there exists an $\omega_{1/2} \in \Omega$ such that for every $z$,
\begin{equation}\label{eq-npc-ineq-midpoint}
    d(z, \omega_{1/2})^2 \leq \frac{1}{2} d(z, \omega_0)^2 + \frac{1}{2} d(z, \omega_1)^2 - \frac{1}{4}d(\omega_0, \omega_1)^2.
\end{equation}
As illustrated in Figure \ref{fig-npc-triangle}, this inequality means that triangles in Hadamard spaces are \textit{thin}. This interpretation can be used equivalently to define of Hadamard spaces via \textit{comparison triangles}, see Chapter 4 of \citet{burago_course_2001}. The following proposition from \citet{auscher_probability_2003} shows that Hadamard spaces are geodesic spaces, and that (\ref{eq-npc-ineq-midpoint}) holds along geodesics.
\begin{proposition}[Proposition 2.3 in \citet{auscher_probability_2003}]\label{prop-sq-dist-convex}
    If $(\Omega, d)$ is a Hadamard space then it is a geodesic space. Moreover, for any pair of points  $\omega_0, \omega_1 \in \Omega$ there exists a unique geodesic connecting them, denoted $\gamma_{\omega_0}^{\omega_1}$. For $t \in [0, 1]$ the intermediate point $\gamma(t)$ depends continuously on the endpoints $\omega_0, \omega_1$. Finally, for any $z \in \Omega$ and $t \in [0, 1]$,
    \begin{equation}\label{eq-npc-ineq}
        d(z, \gamma(t))^2 \leq (1-t) d(z, \gamma(0))^2 + t d(z, \gamma(1))^2 - t(1-t) d(\gamma(0), \gamma(1))^2.
    \end{equation}
\end{proposition}

Since (\ref{eq-npc-ineq-midpoint}) is a special case of (\ref{eq-npc-ineq}), we will also refer to the latter as the NPC inequality. Hadamard spaces and the NPC inequality provide a rich context for the study of random objects.  One important result is that for any $X \in L^1(\Omega)$, the function $\omega \mapsto \expec{d(X, \omega)^2 - d(X, z)^2}$ is continuous and uniformly convex, and hence by completeness of the space, has a unique minimizer, see Proposition 4.3 in \citet{auscher_probability_2003}. Since $z$ only enters through an additive term which does not depend on $\omega$, the minimizer of this function is independent of $z$. This implies an alternative definition of the Fr\'echet mean for any Hadamard-value random variable in $L^1(\Omega)$ via an arbitrary $z \in \Omega$,
\begin{equation*}
    \expec{X} = \argmin_{\omega \in \Omega} \expec{d(X, \omega)^2 - d(X, z)^2}.
\end{equation*}
We mention some further useful results in \ref{sec-app-hadamard} and refer the reader to \citet{auscher_probability_2003} for a thorough review of the subject.

We now present a few examples from \citet{auscher_probability_2003} of Hadamard spaces and ways of building Hadamard spaces out of existing ones. The most well-known case of Hadamard space are Hilbert spaces. Since Functional Data Analysis (FDA) is typically carried out in the $L_2$ Hilbert spaces (see \citet{ramsay_functional_2005}), considering Hadamard spaces allows to approach FDA tasks from a random object perspective.

\begin{example}[Hilbert spaces]
    Let $\c{H}$ be a Hilbert space with induced norm $\norm{\cdot}_{\mathcal{H}}$, then $(\c{H}, d)$ with $d(x, y) = \norm{x - y}_{\cal{H}}$ is a Hadamard space. In Hilbert spaces, the Fr\'echet mean of $X$ corresponds to the Bochner integral with respect to the probability measure $P$ of $X$, $\expec{X} = \int X dP$ (see \citet{hsing_theoretical_2015}). In Hilbert spaces, the NPC inequality (\ref{eq-npc-ineq}) holds to equality.
\end{example}



\begin{example}[Constructed Spaces]\label{ex-constructed} Let $(\Omega, d)$ be a Hadamard space. Then 
    \begin{enumerate}
        \item{Any subset $O \subset \Omega$ is a Hadamard space if and only if it is closed and convex.}
        \item{Let $\Theta$ be an arbitrary set and $\omega : \Theta \rightarrow \Omega$ be a bijection. Then, $\Theta$ is a Hadamard space equipped with the distance $d_\omega(\theta, \theta') = d(\omega(\theta), \omega(\theta'))$. Furthermore, $\expec{X} = \omega^{-1}\left(\expec{\omega(X)}\right)$ holds for any $X \in L^1(\Theta)$.}
    \end{enumerate}
\end{example}

One specific example of Hadamard space of particular interest is the space of one-dimensional density functions over an interval equipped with the 2-Wasserstein distance.

\begin{example}[2-Wasserstein Space]\label{ex-wasserstein}
    Let $\W_2(I)$ be the space of probability measures on $I \subset \R$ with finite second moment. This space, endowed with the 2-Wasserstein distance, is a metric space, see \citet{panaretos_invitation_2020}. Consider the subset $\c{D}(I) \subset \W_2(I)$ of distributions having a density with respect to the Lebesgues measure. For two distributions $\mathbb{P}, \mathbb{Q} \in \c{D}(I)$ with quantile functions $F^{-1}_{\mathbb{P}}, F^{-1}_{\mathbb{Q}}$, the 2-Wasserstein distance between $\mathbb{P}$ and $\mathbb{Q}$ is given by
    \begin{equation*}
        d_{\W_2}(\mathbb{P}, \mathbb{Q}) = \norm{F^{-1}_{\mathbb{P}} - F^{-1}_{\mathbb{Q}}}_{L_2[0,1]}.
    \end{equation*}
    The space of quantile functions being a closed and convex subspace of $L_2[0,1]$, it is also a Hadamard space. Hence, $(\c{D}(I), d_{\W_2})$ falls under the second case described in Example \ref{ex-constructed} and is also a Hadamard space.
\end{example}

Another useful example of a constructed Hadamard space is that of symmetric positive definite (SPD) matrices together with the \textit{Log-Cholesky} distance.

\begin{example}[Log-Cholesky distance]\label{ex-logchol}
    Let $\mathcal{S}_p^+$ be the space of SPD matrices of dimension $p$ and $\mathcal{L}_p^+$ be the space of $p \times p$ lower-triangular matrices with positive diagonal elements. Given a matrix $M \in \mathcal{S}_p^+$, the Cholesky decomposition of $M$ is well defined, meaning that there exists a lower-triangular matrix with positive diagonal elements $L \in \mathcal{L}_p^+$ such that $M = LL^\top$. Let $\floor{M}$ be the $p \times p$ matrix such that $\floor{M}_{ij} = M_{ij}$ if $i < j$ and $0$ otherwise and $D(M)$ be the $p \times p$ diagonal matrix with diagonal entries $D(M)_{ii} = M_{ii}$. While simply using the Froebenius distance between Cholesky factors of SPD matrices yields a valid distance, \citet{lin_riemannian_2019} argues that it leads to an unwanted \textit{swelling} effect in geodesics and proposes another distance $d_{\mathcal{S}_p^+}$ treating the diagonal and strictly lower triangular parts of $L$ differently. Let $M_1, M_2 \in \mathcal{S}_p^+$ with Cholesky factors $L_1$ and $L_2$, then the distance $d_{\mathcal{S}_p^+}(M_1, M_2)$ is given by
    \begin{equation*}
        d_{\mathcal{S}_p^+}(M_1, M_2)^2 = \norm{\floor{L_1} - \floor{L_2}}_{F}^2 + \norm{\log D(L_1) - \log D(L_2)}_{F}^2,
    \end{equation*}
    where $\norm{\cdot}_F$ is the Froebenius norm. As a case of Example \ref{ex-constructed}, this is a Hadamard space, which is also shown in \citet{lin_riemannian_2019}, together with other properties of the space.
\end{example}
\section{The GAR(1) Model}\label{sec-model}

\subsection{Model and stationary solution}

Let us consider first a time series $\Xts$ in $\R$ with constant mean $\expec{X_t} = \mu$ for all $t \in \N$. Then, $\Xts$ follows a first-order autoregressive model, denoted AR(1), with concentration parameter $\varphi$ if it satisfies the following relation
\begin{equation}\label{eq-std-ar1}
    X_{t+1} - \mu = \varphi(X_t - \mu) + \varepsilon_{t+1},
\end{equation}
where the noise terms $\eset{\varepsilon_t}_{t \in \N}$ are i.i.d random variables with mean 0. Without the structure of a vector space, this model cannot be directly formulated in general metric spaces. A key insight towards a more general definition of AR(1) models is that (\ref{eq-std-ar1}) can be rewritten as
\begin{equation*}
    X_{t+1} = (1 - \varphi) \mu + \varphi X_t + \varepsilon_{t+1}.
\end{equation*}
This shows that each random variable of the time series can be written as a weighted sum of the overall mean of the time series and the previous observation, perturbed by a centered random noise. For $\varphi \in [0,1]$, this weighted sum corresponds to the point along the geodesic from $\mu$ to $X_t$ at $\varphi$. This interpretation can be used to define an autoregressive process only using geodesics.

Let now $(\Omega, d)$ be a Hadamard space. In this context, we consider a broad class of noise models represented by random maps $\varepsilon : \Omega \rightarrow \Omega$. We say that a random map $\varepsilon$ is \textit{unbiased} if for all $\omega \in 
\Omega$, the random variable $\varepsilon(\omega)$ is in $L^1(\Omega)$ and $\expec{\varepsilon(\omega)} = \omega$. Note that the expectation is the Fr\'echet mean in $\Omega$, hence by (\ref{eq-frechet-def}), the previous statement can be written as
\begin{equation} \label{eq-unbiased}
    \expec{d(\varepsilon(\omega), \omega)^2 - d(\varepsilon(\omega), \omega')^2} < 0\qquad\t{ for all } \omega, \omega' \in \Omega\,\t{with}\,\omega \neq \omega'.
\end{equation}

\begin{figure}[t!]
    \centering
    \includegraphics[width=0.5\textwidth]{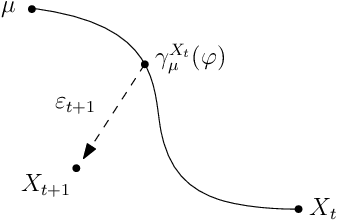}
    \caption{Illustration of the iterated equation (\ref{eq-iterated}).}
    \label{fig-gar}
\end{figure}

We say that a sequence of random variables $\Xts \subset L^1(\Omega)$ with common mean $\mu$ follows the \textit{geodesic autoregressive} model of order 1, GAR(1), with concentration parameter $\varphi \in [0,1]$ if it satisfies the following iterated system of equations 
\begin{equation}\label{eq-iterated}
    X_{t+1} = \varepsilon_{t+1} (\gamma_\mu^{X_t}(\varphi)),
\end{equation}
where $\eset{\varepsilon_t}_{t\in \N}$ are i.i.d unbiased noise maps and $\gamma_\mu^{X_t}$, we recall, is the (random) geodesic connecting $\mu$ to $X_t$. The data generating process is illustrated in Figure \ref{fig-gar}.


    This relation opens the question of whether the condition of a shared Fr\'echet mean $\expec{X_t} = \mu$ and equation (\ref{eq-iterated}) can mutually be fulfilled. Using that $\eset{\varepsilon_t}_{t \in \N}$ are unbiased and assuming that $\expec{X_t} = \mu$, the condition becomes
\begin{equation*}
    \expec{\gamma_{\expec{X_t}}^{X_t}(\varphi) } = \expec{X_t}.
\end{equation*}
Unfortunately, this condition does not hold in every Hadamard space. Thus, we will assume the following.

\begin{assumption}\label{ass-mean-preserving}
    For every $X \in L^1(\Omega)$ and $\varphi \in [0,1]$, $\expec{\gamma_{\expec{X}}^X(\varphi)} = \expec{X}$.
\end{assumption}
Assumption \ref{ass-mean-preserving} holds for some of the examples mentioned in the previous section. For any Hilbert space $\c{H}$, this condition holds by linearity of the expectation since for any $X \in L^1(\c{H})$,
\begin{equation*}
    \expec{\gamma_{\expec{X}}^X(\varphi)} = \expec{(1 - \varphi)\expec{X} + \varphi X} = (1 - \varphi)\expec{X} + \varphi \expec{X} = \expec{X}.
\end{equation*}
Furthermore, if the condition holds for a Hadamard space $(\Omega, d)$, then it also holds for a Hadamard space constructed by taking the image of a bijection $\omega$ as described in the second part of Example \ref{ex-constructed}.
\begin{lemma} \label{lem-ass-constructed}
    Let $(\Omega, d)$ be a Hadamard space and $(\Theta, d_\omega)$ be a constructed Hadamard space based on the bijection $\omega$. If $(\Omega, d)$ satisfies Assumption \ref{ass-mean-preserving}, then so does $(\Theta, d_\omega)$.
\end{lemma}
Lemma \ref{lem-ass-constructed} shows that Assumption \ref{ass-mean-preserving} holds for a large class of Hadamard spaces, in particular for the subspace $\c{D}(I)$ of $\W_2(I)$ of distributions having a density function, as described in Example \ref{ex-wasserstein}.

To show the existence of a stationary solution to Equation (\refeq{eq-iterated}), we use the framework of \textit{iterated random function systems} presented in \citet{wu_limit_2004}. Let us first introduce some notation. Given the i.i.d noise maps $\eset{\varepsilon_t}_{t\in \N}$, define for all $t \in \N$ the random functions $F_t : \Omega \rightarrow \Omega$,
\begin{equation*}
    F_t(x) = \varepsilon_{t}(\gamma_\mu^x(\varphi)).
\end{equation*}
Then, Equation (\ref{eq-iterated}) can be rewritten as an iterated random function system,
\begin{equation*}
    X_{t+1} = F_{t+1}(X_t).
\end{equation*}
Further, for any $t \in \N$ and $x \in \Omega$, the following random variable will be useful in expressing the condition of existence of a stationary solution,
\begin{equation}\label{eq-f-def}
    X_t(x) =  F_t \circ F_{t-1} \circ \ldots \circ F_1(x),
\end{equation}
then, $X_t = X_t(X_0)$, and the construct also allows to study a coupled version $X_t(X_0')$ of $X_t$ for $X_0' \overset{\mathcal{D}}{=} X_0$. The following theorem provides a sufficient condition for the existence of a stationary solution to (\ref{eq-f-def}) based on a \textit{geometric-moment contracting} condition on the iteration maps $\eset{F_t}_{t \in \N}$.

\begin{theorem}[Theorem 2 of \citet{wu_limit_2004}]\label{thm-wu}
    Suppose there exists $x_0 \in \Omega, \alpha > 0, r \in (0, 1)$ and $C > 0$ such that
    \begin{equation}\label{eq-geom-cond}
        \expec{d(X_t(x), X_t(x_0))^\alpha} \leq C r^t d(x, x_0)^\alpha
    \end{equation} 
    holds for all $x \in \Omega, t \in \N$. Then, for all $x \in \Omega$
    \begin{equation*}
        X^\star_t = \lim_{m \rightarrow \infty} F_t \circ F_{t-1} \circ \ldots \circ F_{t - m + 1}(x)
    \end{equation*}
    exists and does not depend on $x$. Moreover, $\eset{X^\star_t}_{t \in \N}$ is a stationary solution of Equation (\ref{eq-iterated}).
\end{theorem}

As noted in \citet{wu_limit_2004}, if condition \refeq{eq-geom-cond} holds for some $\alpha \geq 1$, then Hölder's inequality can be directly use to show that it also holds for any $\beta \in (0, \alpha)$. While the GAR(1) model can be defined for $\varphi = 1$, it is unlikely that the process will have a stationary solution in this case. This is due to the fact that the noise maps $\eset{\varepsilon_t}_{t \in \N}$ are unbiased, and the iterated system of equations will not converge to a stationary solution if the noise maps are not contracting.

If we assume that the noise maps $\eset{\varepsilon_t}_{t\in\N}$ are Lipschitz with random Lipschitz constants $K_t \in L^\alpha(\R_{+})$ with $K := \expec{K_t^\alpha}$, we have by using the Geodesic Comparison Inequality (see \ref{prop-geo-comp}) on $d(\gamma_\mu^x (\varphi), \gamma_\mu^{x_0} (\varphi))^2$
\begin{align*}
    \expec{d(X_1(x), X_1(x_0))^\alpha}
    &= \expec{d(\varepsilon_1(\gamma_\mu^x(\varphi)), \varepsilon_1(\gamma_\mu^{x_0}(\varphi)))^\alpha}\\
    &\leq K d(\gamma_\mu^x (\varphi), \gamma_\mu^{x_0} (\varphi))^\alpha\\
    &\leq K \varphi^{\alpha/2} d(x, x_0)^\alpha.
\end{align*}
By induction, this implies $\expec{d(X_t(x), X_t(x_0))^\alpha} \leq \left(K \varphi^{\alpha/2}\right)^t d(x, x_0)^\alpha$, which shows that condition (\ref{eq-geom-cond}) holds if $r = K \varphi^{\alpha/2} < 1$.
\subsection{Identifiability}

Under the conditions of Theorem \ref{thm-wu}, Equation (\ref{eq-iterated}) has a stationary solution and the model features two quantities of interest: the time-invariant Fr\'echet mean of the time series $\mu \in \Omega$, and the concentration parameter $\varphi \in [0, 1]$. Before considering the estimation of these quantities, we show that both are identifiable. The identifiability of the Fr\'echet mean follows directly from the stationarity of the time series and the definition and existence of the Fr\'echet mean in a Hadamard space. 

\begin{theorem}
    Let $\Xts \subset \Omega$ and assume that $\Xts$ satisfies the conditions of Theorem \ref{thm-wu}. Then, the Fr\'echet mean $\mu = \expec{X_t}$ is identifiable.
\end{theorem}

As for the concentration parameter, we can consider the mean squared error  
\begin{equation}\label{eq-def-L}
    L(u) = \expec{d\left(X_{t+1}, \gamma_{\mu}^{X_t}(u)\right)^2}.
\end{equation}
Then, assuming that the noise maps are unbiased, we can show that this loss is uniquely minimized by the true concentration parameter $\varphi$.

\begin{theorem}\label{thm-phi-id}
    Let $\Xts \subset \Omega$, assume that $\Xts$ is in $L^2(\Omega)$ and satisfies Equation (\ref{eq-iterated}) with true concentration parameter $\varphi \in [0,1]$. Assume further that the noise maps $\eset{\varepsilon_t}_{t\in \N}$ are unbiased. Then, $\varphi$ is the unique minimizer of $L$.
\end{theorem}
\begin{proof}
    Since $\gamma_{\mu}^{X_t}(\varphi)$ is the Fr\'echet mean of $X_{t+1}$ given $X_t$ and $\varepsilon_{t+1}$ is unbiased, for all $\varphi' \in [0,1]$, with $\varphi' \neq \varphi$, we have 
    \begin{align*}
        L(\varphi') 
        &= \expec{d\left(X_{t+1}, \gamma_{\mu}^{X_t}(\varphi')\right)^2}
        = \expec{d\left(\varepsilon_{t+1}(\gamma_{\mu}^{X_t}(\varphi)), \gamma_{\mu}^{X_t}(\varphi')\right)^2}\\
        &= \expec{\expec{d\left(\varepsilon_{t+1}(\gamma_{\mu}^{X_t}(\varphi)), \gamma_{\mu}^{X_t}(\varphi')\right)^2\mid X_t}}\\
        &> \expec{\expec{d\left(\varepsilon_{t+1}(\gamma_{\mu}^{X_t}(\varphi)), \gamma_{\mu}^{X_t}(\varphi)\right)^2\mid X_t}}
        = L(\varphi)
    \end{align*}
    where the inequality follows from the unbiasedness of $\varepsilon_{t+1}$ as defined in Equation \eqref{eq-unbiased}.
\end{proof}

\section{Estimation of model parameters}\label{sec-estimation}

Now that the identifiability of the Fr\'echet mean and the concentration parameter have been established, we show that empirical estimation of the associated risks produces consistent estimators. Furthermore, we show that the Fr\'echet mean can be estimated at a $\sqrt{T}$-rate.

\subsection{Fr\'echet mean}

For simplicity, we assume that the $\Xts$ are $L^2(\Omega)$. Then, the Fr\'echet function $M(\omega) = \expec{d(X, \omega)^2}$ has a natural empirical version based on the observations $X_1, \ldots, X_T$.
\begin{equation}\label{eq-sample-frechet-fn}
    M_T(\omega) = \frac{1}{T}\sum_{t=1}^T d(X_t, \omega)^2.
\end{equation}
We define the estimator $\hat\mu_T$ as the minimizer of $M_T$, which is well-defined by strict convexity of the squared distance in Hadamard spaces. The asymptotic behavior of $\hat\mu_T$ is described by the theory of M-estimation, see for instance \citet{van_der_vaart_weak_1996}, where consistency and rates of convergence are readily available for i.i.d data. Here, we adapt results on iterated random function system from \citet{wu_limit_2004} to verify the general assumptions for M-estimation presented in \citet{van_der_vaart_weak_1996}. One assumption which is standard in the study of random objects, concerns the covering number of $(\Omega, d)$.

\begin{assumption} \label{ass-cov-num}
    Let $B(\mu, \delta)$ be the ball in $\Omega$ of size $\delta$ centered in $\mu$ and $N(\varepsilon, B_\delta(\mu))$ be the covering number of $B_\delta(\mu)$ using balls of size $\varepsilon$. Assume
    \begin{equation*}
        \int_0^1 \sqrt{1 + \log N(\varepsilon\delta, B_\delta(\mu))}\, \d \varepsilon = O(1)\qquad \t{ as } \delta \rightarrow 0.
    \end{equation*}
\end{assumption}

In the following theorem, we show that this assumption, together with the assumptions required for stationary of the sequence $\Xts$, are enough to obtain the $\sqrt{T}$ consistency of the mean estimator. Note that this result is of more general interest since it does not assume that the data follows our GAR(1) model but only requires control on the dependency of the sequence $\eset{X_t}_{t \in \N}$.

\begin{theorem}\label{thm-rate-mu}
    Let $(\Omega, d)$ be a Hadamard space and $\eset{X_t}_{t \in \N}$ be an $L^2(\Omega)$ sequence of random variables satisfying Equation (\ref{eq-geom-cond}) for some $\alpha \geq 1$. Suppose that Assumption \ref{ass-cov-num} holds around the Fr\'echet mean $\mu = \expec{X_t}$. Then, the minimizer $\hat\mu_T$ of $M_T$ is a consistent estimator of $\mu$ and satisfies
    \begin{equation}\label{eq-rate-mu-hat}
        \sqrt{T} d(\mu, \hat\mu_T) = O_P(1).
    \end{equation}
\end{theorem}

The proof of Theorem \ref{thm-rate-mu} can be found in \ref{sec-proofs}.

\subsection{Concentration parameter}

Similarly to the Fr\'echet mean, we construct an estimator of the concentration parameter by minimizing an empirical version of $L$ in Equation (\ref{eq-def-L}). We estimate the expectation with the available sample and replace the Fr\'echet mean $\mu$ by the estimator $\hat\mu_T$, giving the following risk function
\begin{equation}\label{eq-def-Ln}
    L_T(u) = \frac{1}{T-1}\sum_{t=1}^{T-1} d(X_{t+1}, \gamma_{\hat\mu_T}^{X_t}(u))^2.
\end{equation}

We prove the consistency of the resulting estimator based on results from \citet{newey_uniform_1991} relying on the compactness of the domain $[0,1]$ and continuity results about $L$ and $L_T$. The consistency result is the following.

\begin{theorem}\label{thm-phi-consistent}
    Assume that $(\Omega, d)$ and $\Xts$ satisfy the conditions of Theorem \ref{thm-rate-mu}. Then, the minimizer $\hat\varphi_T$ of $L_T$ is a consistent estimator of $\varphi$.
\end{theorem}
\begin{proof}
    By Proposition \ref{prop-uniform}, we have that $\norm{L_T - L}_\infty = o_P(1)$. Together with the identifiability result in Theorem \ref{thm-phi-id}, we have that $L$ has a unique minimizer. By Corollary 3.2.3 in \citet{van_der_vaart_weak_1996}, any sequence of minimizers $\hat\varphi_T$ of $L_T$ satisfies $\abs{\hat\varphi_T - \varphi} = o_P(1)$.
\end{proof}

Furthermore, we show in Lemma \ref{lem-L-strongly-convex} in the Appendix that $L$, and hence $L_T$, is strongly convex. This makes it possible to use generic convex solvers to find the minimizer of $L_T$.


\section{Test for serial independence}\label{sec-hypothesis-test}


One hypothesis test of interest is whether the random variables $\Xts$ are independent, which corresponds in the GAR(1) model to testing $H_0 : \varphi = 0\,\t{vs.}\,H_1 : \varphi > 0$. Since no strong results are available about the asymptotic distribution of $\hat\varphi_T$, another test statistic must be considered. To that end, let us consider the statistic
\begin{equation}\label{eq-def-Dn}
    D_T = \frac{1}{T-1}\sum_{t=1}^{T-1} d(X_t, X_{t+1})^2.
\end{equation}
We proceed to show that $D_T$ is asymptotically normal with mean and variance depending on the value of the concentration parameter $\varphi$, with smaller values of $D_T$ taken for larger values of $\varphi$. This allows us to build a test that asymptotically has correct level and power.

\begin{theorem}\label{thm-h0-clt}
    Let $X_1, X_2, X_3$ be i.i.d copies of $\varepsilon(\mu)$, then, under $H_0$ and as $T \rightarrow \infty$,
    \begin{equation*}
        \sqrt{T}(D_T - \expec{d(X_1, X_2)^2}) \rightarrow N(0, \sigma_0^2),
    \end{equation*}
    where $\sigma_0^2 = \Var{d(X_1, X_2)^2} + 2\Cov{d(X_1, X_2)^2, d(X_1, X_3)^2}$.
\end{theorem}
\begin{proof}
    Under $H_0$, the sequence $\Xts$ is formed of i.i.d random variables. We consider the centered summands, $Y_t = d(X_t, X_{t+1})^2 - \expec{d(X_1, X_2)^2}$. The sequence $\eset{Y_t}_{t \in \N}$ is then $m$-dependent with $m = 1$. We obtain the desired result by the Central Limit Theorem for $m$-dependent sequences, see Theorem 2 in \citet{hoeffding_central_1948}.
\end{proof}

To study the behavior of this test statistic under $H_1 : \varphi \neq 0$, we base our analysis on Theorem 3 of \citet{wu_limit_2004} which provides conditions for the asymptotic normality of sums of the form of $D_T$. Under the assumptions required for the existence of a stationary solution, we obtain the result.
\begin{theorem}
    Assume that $\Xts$ satisfies the conditions of Theorem \ref{thm-wu} with $\varphi > 0$, then there exists a $\sigma_{\varphi} \geq 0$ such that
    \begin{equation*}
        \sqrt{T}(D_T - \expec{d(X_1,
         X_2)^2}) \rightarrow N(0, \sigma_{\varphi}^2).
    \end{equation*}
\end{theorem}

In general, it is not clear whether $\expect{\varphi = 0}{D_T} \neq \expect{\varphi = \varphi^\star}{D_T}$ for an arbitrary $\varphi^\star \neq 0$, and whether the test described above has asymptotic power. One possible way to avoid this issue is to require the following monotonicity condition on the noise maps.

\begin{assumption} \label{ass-noise-monotonicity}
    For all $x, y, z \in \Omega$, then the noise maps $\varepsilon$ satisfy the following monotonicity condition
    \begin{equation}\label{eq-noise-monotonous}
        d(x, z) < d(y, z) \Rightarrow \expec{d(\varepsilon(x), z)^2} < \expec{d(\varepsilon(y), z)^2}.
    \end{equation}
\end{assumption}

For any $\varphi > 0$, $d(\gamma_\mu^{X_t}(\varphi), X_t) = (1 - \varphi)d(\mu, X_t) < d(\mu, X_t)$. Together with Assumption \ref{ass-noise-monotonicity}, this gives
\begin{equation*}
    \expect{\varphi = \varphi^\star}{D_T} = \expec{d(\varepsilon_{t+1}(\gamma_\mu^{X_t}(\varphi^\star)), X_t)^2} < \expec{d(\varepsilon_{t+1}(\mu), X_t)^2} = \expect{\varphi = 0}{D_T},
\end{equation*}
which implies that the asymptotic power of the test is 1.

To construct a level $\alpha$ hypothesis test for $H_0 : \varphi = 0\,\t{vs.}\,H_1 : \varphi > 0$, one could reject $H_0$ if the absolute deviation of $D_T$ from its asymptotic mean exceeds a certain threshold $q_\alpha$ based on the result in Theorem \ref{thm-h0-clt}. However, the asymptotic mean and variance of $D_T$ required for this test depend on the underlying data distribution and are unknown. Alternatively, one could attempt to center $D_T$, for instance by considering the randomized statistic $\tilde D_T = \frac{1}{T-1} \sum_{t=1}^{T-1} d(X_t, X_{t+1})^2 - d(X_t, X_{\pi(t)})^2$, where $\pi$ is a random permutation. Similar theoretical arguments as for $D_T$ shows that $\tilde D_T$ is asymptotically normal under $H_0$, with zero mean and a variance estimable from data. This enables normalization of $\tilde D_T$ and the construction of a test based on the asymptotic approximation.

Instead, we use a permutation procedure to compute approximate $p$-values under $H_0$ for better finite sample properties. Specifically, let $B \in \N$ be the number of permutations used for constructing the approximate $p$-value and let $\pi_1, \ldots, \pi_B$ be random permutations of $\eset{1, \ldots, T}$. For each permutation $\pi$, we denote by $D_T^\pi$ the test statistic computed based on the permuted sample $\left(X_{\pi(1)}, \ldots, X_{\pi(T)}\right)$ and define the approximated $p$-value, $\hat p_B = \frac{1}{B}\sum_{b=1}^B \mathbbm{1}\eset{D_T \geq D^{\pi_b}_T}$. Under the null hypothesis of independence, we have that $D_T^\pi \overset{\mathcal{D}}{=} D_T$, and shuffling the observations under the alternative allows to loosen the dependency between consecutive observations, giving an approximate sample under $H_0$. The resulting level $\alpha$ test is then constructed by rejecting $H_0$ if $\hat p_B \leq \alpha$, see \citet{hemerik_exact_2018}.


\section{Numerical experiments}\label{sec-numerical}

In the following, we illustrate our theoretical results with numerical experiments taking place in different Hadamard spaces. We empirically verify the convergence rate of $\hat\mu_T$ proved in Theorem \ref{thm-rate-mu}, verify the consistency of $\hat\varphi_T$ proved in Theorem \ref{thm-phi-consistent} and show that the test constructed via the bootstrapping procedure described in Section \ref{sec-hypothesis-test} has the desired size and increasing power as $T$ grows.

We study three scenarios of time series following the GAR(1) model. The first example is that of the real line $\R$ equipped with the standard Euclidean distance, with a multiplicative noise model. For the second example, we consider the space of density distributions over the real line equipped with the 2-Wasserstein distance, with a geodesic noise model. For the last example, we consider SPD matrices with the Log-Cholesky metric with a noise model based on a Lie group structure defined in \citet{lin_riemannian_2019}.

In each of these scenarios, we generate time series of different lengths $T \in \eset{40,80,160,320,640}$ and for different values of the concentration parameter $\varphi \in \eset{0, 0.1, 0.2, 0.3, 0.4, 0.5, 1}$. Naturally, $\varphi = 0$ and $\varphi = 1$ are special cases that we will consider with care in the evaluation of our results. For each combination of metric space, number of observations and concentration parameter, we generate 1000 datasets. For each dataset, we compute the estimators $\hat\mu_T$ and $\hat\varphi_T$, and run the permutation-based hypothesis test at level $\alpha = 0.05$. We report for each combination of a metric space, $T$ and $\varphi$, the average estimation errors $d(\hat\mu_T, \mu)$ and $\abs{\hat\varphi_T - \varphi}$, as well as properties of the hypothesis test, all calculated over the 1000 runs. 

Additional results can be found in \ref{sec-app-add-sims} where we compare our permutation test for $D_T$ to three other tests: another permutation test based on the $\hat\varphi_T$ estimator and the two tests developed in \citet{jiang_testing_2023}. This appendix shows that the test based on $D_T$ is better calibrated and achieves considerably higher power than all other evaluated in all simulation scenarios considered.

All simulations and analyses are done in Python. The code to reproduce the experiments and figures is available online\footnote{\url{https://github.com/matthieubulte/GAR}}.

\begin{figure}[t!]
    \begin{tikzpicture}
        \draw[dotted, rounded corners=4pt] (-4.2, 3.2) -- (9, 3.2) -- (9,-.5) -- (-4.2,-.5) -- cycle;
        \node (A) at (2.25, 1.5) {\includegraphics[width=11cm]{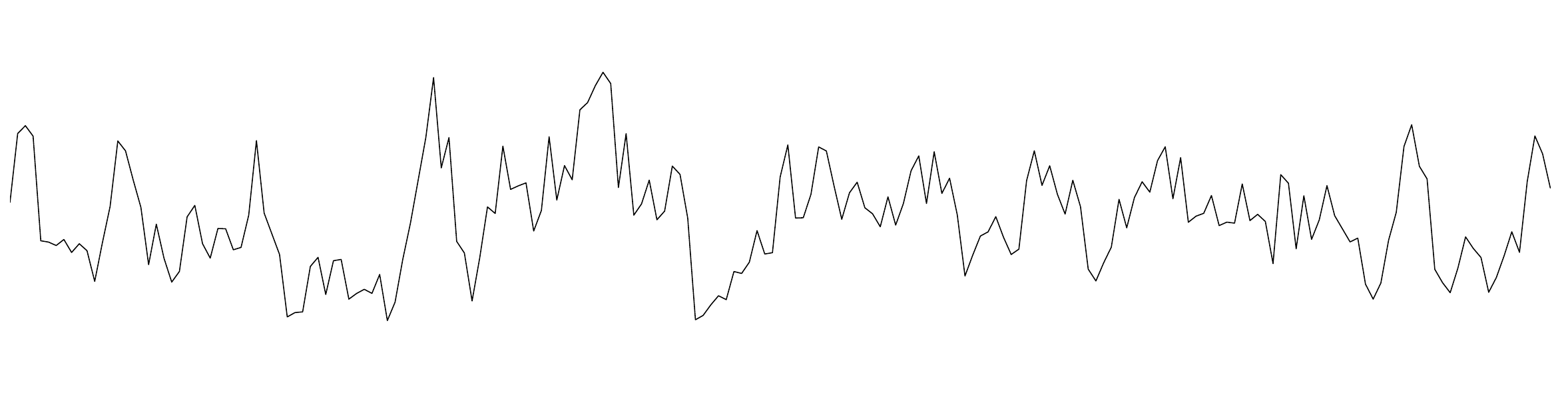}};
        \node[fill=white] (aa) at (-4, 3) {\small\textbf{a}};
        \draw (-3.8, 3.2) -- ++(0,-.4)  --  ++(-.4,0) {[rounded corners=4pt]-- ++(0,.4)} -- cycle;

        \node (x_start) at (-3.5, .05) {};
        \node (x_end) at (8, .05) {};
        \draw[->, line width=0.25mm, axis_color] (x_start) edge (x_end);
        \node (t) at (2.25, -.2) {Time};

        \draw[dotted, rounded corners=4pt] (-4.2, -.6) -- (0.1, -.6) -- (0.1,-4.6) -- (-4.2,-4.6) -- cycle;
        \node (B) at (-2.05, -2.65) {
            \includegraphics[width=4cm]{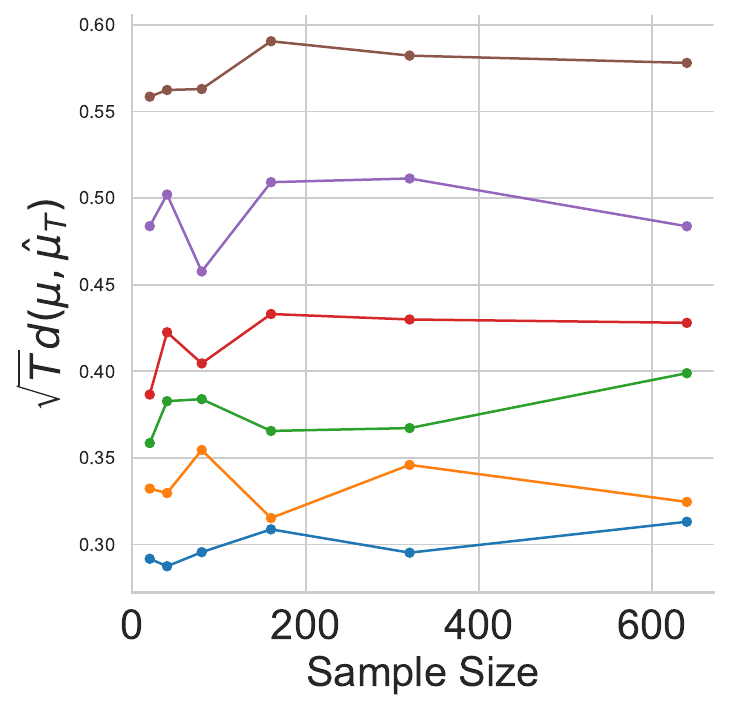}
        };
        \node[] (bb) at (-4, -.8) {\small\textbf{b}};
        \draw (-3.8, -.6) -- ++(0,-.4)  --  ++(-.4,0) {[rounded corners=4pt]-- ++(0,.4)} -- cycle;
        
        \draw[dotted, rounded corners=4pt] (.2, -.6) -- (4.5, -.6) -- (4.5,-4.6) -- (.2,-4.6) -- cycle;
        \node (C) at (2.35, -2.6) {\includegraphics[width=4cm]{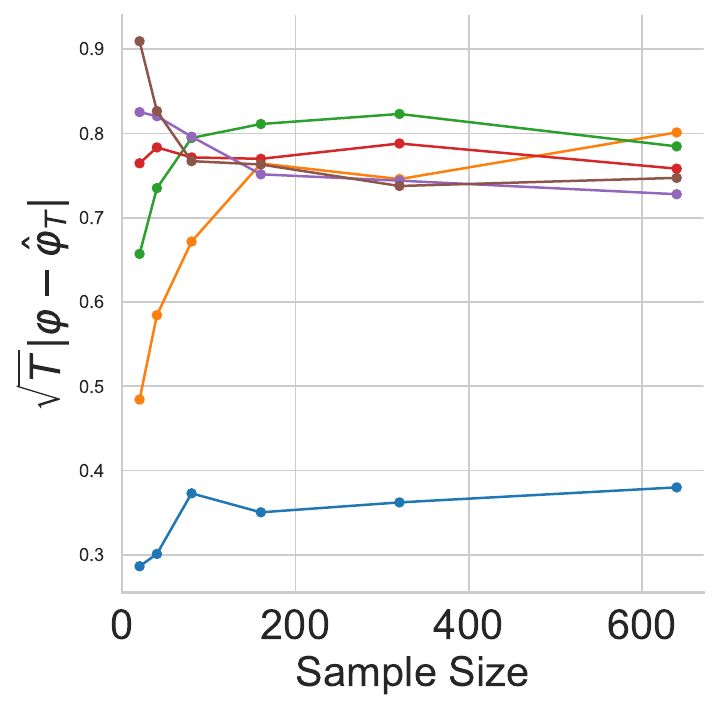}};
        \node[] (cc) at (.4, -.8) {\small\textbf{c}};
        \draw (.6, -.6) -- ++(0,-.4)  --  ++(-.4,0) {[rounded corners=4pt]-- ++(0,.4)} -- cycle;

        \draw[dotted, rounded corners=4pt] (4.6, -.6) -- (9, -.6) -- (9,-4.6) -- (4.6,-4.6) -- cycle;
        \node (D) at (6.75, -2.65) {\includegraphics[width=4cm]{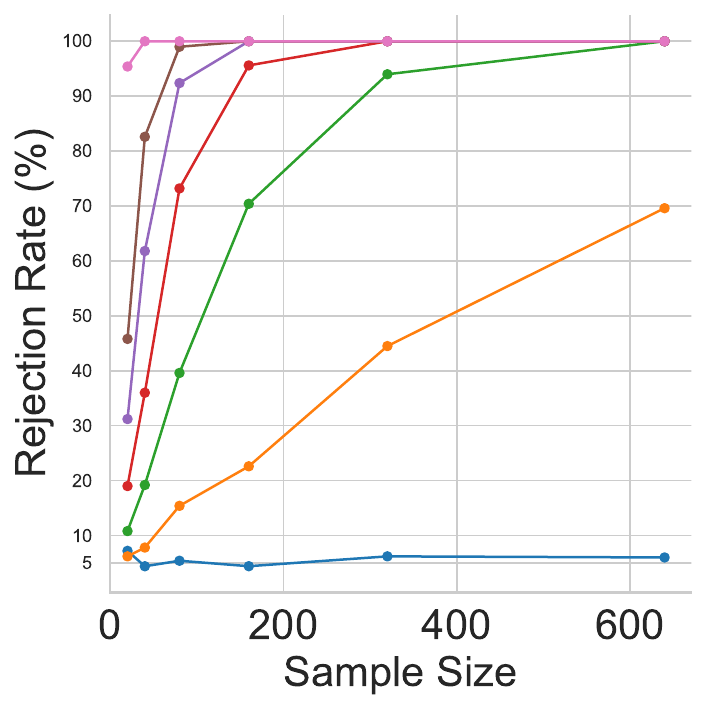}};
        \node[] (dd) at (4.8, -.8) {\small\textbf{d}};
        \draw (5, -.6) -- ++(0,-.4)  --  ++(-.4,0) {[rounded corners=4pt]-- ++(0,.4)} -- cycle;

        \draw[dotted, rounded corners=4pt] (-1.15, -4.7) -- ++(7, 0) -- ++(0,-.55) -- ++(-7,0) -- cycle;
        \node (L) at (2.35, -5.1) {\includegraphics[width=7cm]{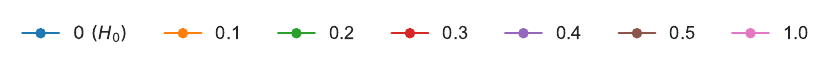}};
        \node[] (ll) at (2.35, -4.9) {\tiny$\varphi$};
    \end{tikzpicture}   

    \caption{Panel (a) depicts a trajectory of 200 time steps from the GAR(1) process described in Section \ref{subsec-num-r}. Panel (b) illustrates the $\sqrt{T}$ convergence of the mean estimator $\hat\mu_T$. Panel (c) shows the convergence of the concentration parameter estimator $\hat\varphi_T$.
    Panel (d) shows the rejection rate of the independence test described in Section \ref{sec-hypothesis-test} with target level 0.05 for different values of $\varphi$.
    }
    \label{fig:r}
\end{figure}

\subsection{$\R$ with multiplicative noise}\label{subsec-num-r}

In the first experimental setup, we investigate the simple case of $(\Omega, d)$ being the real line line $\R$ equipped with the Euclidean distance $d(x, y) = \abs{x - y}$. Here, the Fr\'echet mean corresponds to the ordinary mean and geodesics are given by straight lines, $\gamma_x^y(t) = (1 - t)x + ty$. Furthermore, Assumption \ref{ass-cov-num} is verified since $N(\varepsilon\delta, B_\delta(\mu)) = \varepsilon^{-1}$, thus the entropy integral is bounded and does not depend on $\delta$.

One can see that the unconstrained minimizer of Equation (\refeq{eq-def-Ln}) is the sample autocorrelation. Using the convexity of $L_T$, the estimator $\hat \varphi_T$ of the concentration parameter is still available in closed form by clipping the autocorrelation to positive numbers, giving
\begin{equation}\label{eq-est-r-varphi}
    \hat \varphi_T = \max \eset{0, \frac{\sum_{t=1}^{T-1} (X_{t+1} - \bar X)(X_{t} - \bar X)}{\sum_{t=1}^{T-1} (X_t - \bar X)^2}},
\end{equation}
where $\bar X$ is the sample mean of the time series.

We consider multiplicative noise maps $\varepsilon_i(x) = (1 + \eta_i)x$ where $\eta_i \sim N(0, \sigma^2)$. Then, the noise maps $\eset{\varepsilon_t}_{t \in \N}$ are unbiased and the condition of Theorem \ref{thm-wu} is satisfied for $\varphi < (1 + \sigma^2)^{-1/2}$ since
\begin{equation*}
    \expec{(X_t(x_0) - X_t(x))^2} = \left[\varphi^2(1 + \sigma^2)\right]^t(x_0 - x)^2.
\end{equation*}
In our simulation setup, we work with $\sigma^2 = 0.25^2$ which gives an upper bound $\varphi < (1 + \sigma^2)^{-1/2} \approx 0.97$.

The theoretical results presented in the previous sections are illustrated in Figure \ref{fig:r}. The $\sqrt{T}$ convergence of the Fr\'echet mean estimator holds for $\varphi < 1$ with a reasonably stable value at the tested sample sizes, see panel (b). Panel (c) indicates that the estimator $\hat\varphi_T$ also converges at the parametric $\sqrt{T}$ rate to $\varphi$ for $\varphi < 1$. We can see that for $\varphi = 0$, the error is lower which can be explained by the fact the the sample is i.i.d.\,in this case. In the non-i.i.d.\,case, the estimation error seems smaller for the larger values of $\varphi$ considered. Finally, the rejection rates for the independence test presented in Section \ref{sec-hypothesis-test} show that the test is well calibrated and achieves high power for moderate sample sizes, see panel (d). 
\subsection{Univariate distributions with a density}\label{subsec-wasserstein}

\begin{figure}[ t]
    \begin{tikzpicture}
        \draw[dotted, rounded corners=4pt] (-4.2, 3.2) -- (9, 3.2) -- (9,-.5) -- (-4.2,-.5) -- cycle;
        \node (A) at (2.25, 1.7) {\includegraphics[width=12.5cm]{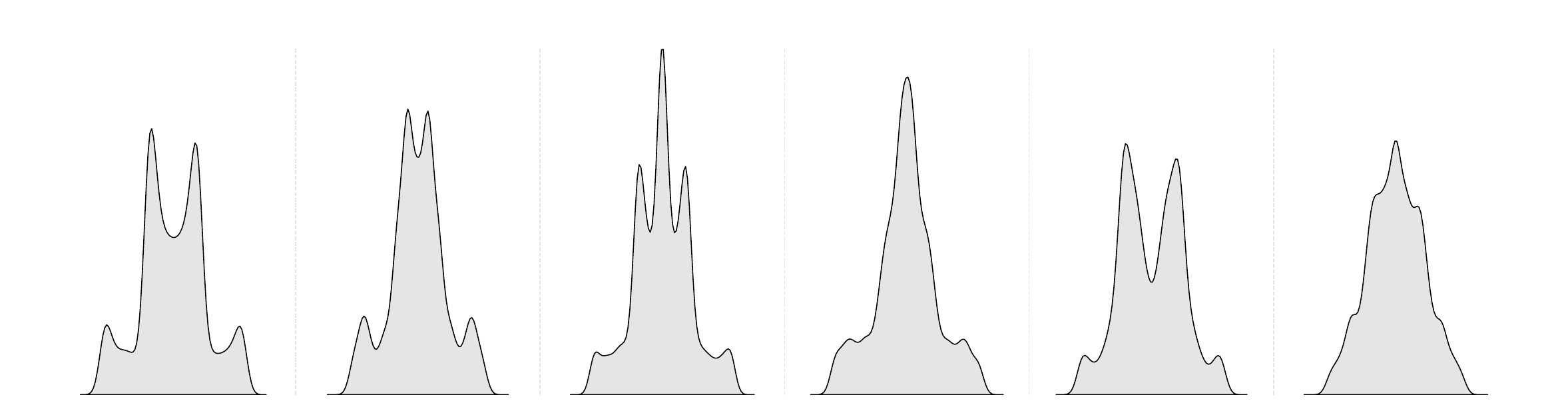}};
        \node[fill=white] (aa) at (-4, 3) {\small\textbf{a}};
        \draw (-3.8, 3.2) -- ++(0,-.4)  --  ++(-.4,0) {[rounded corners=4pt]-- ++(0,.4)} -- cycle;

        \node (x_start) at (-3.5, .05) {};
        \node (x_end) at (8, .05) {};
        \draw[->, line width=0.25mm, axis_color] (x_start) edge (x_end);
        \node (t) at (2.25, -.2) {Time};

        


        \draw[dotted, rounded corners=4pt] (-4.2, -.6) -- (0.1, -.6) -- (0.1,-4.6) -- (-4.2,-4.6) -- cycle;
        \node (B) at (-2.05, -2.65) {
            \includegraphics[width=4cm]{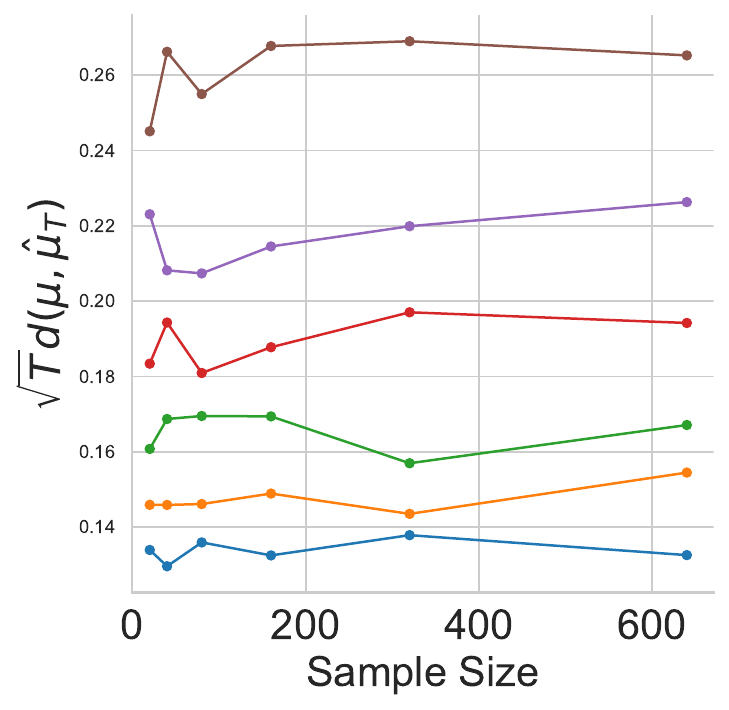}
        };
        \node[] (bb) at (-4, -.8) {\small\textbf{b}};
        \draw (-3.8, -.6) -- ++(0,-.4)  --  ++(-.4,0) {[rounded corners=4pt]-- ++(0,.4)} -- cycle;
        
        \draw[dotted, rounded corners=4pt] (.2, -.6) -- (4.5, -.6) -- (4.5,-4.6) -- (.2,-4.6) -- cycle;
        \node (C) at (2.35, -2.6) {\includegraphics[width=4cm]{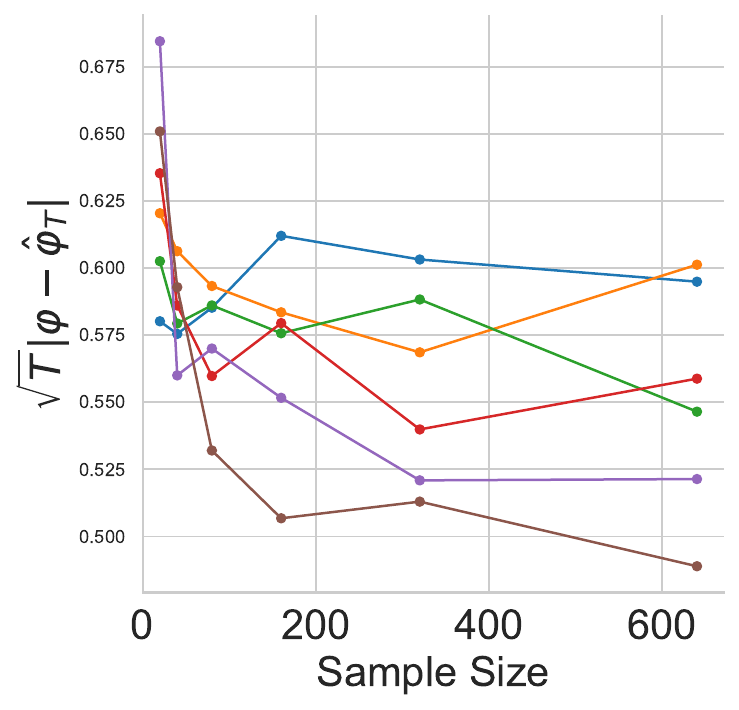}};
        \node[] (cc) at (.4, -.8) {\small\textbf{c}};
        \draw (.6, -.6) -- ++(0,-.4)  --  ++(-.4,0) {[rounded corners=4pt]-- ++(0,.4)} -- cycle;

        \draw[dotted, rounded corners=4pt] (4.6, -.6) -- (9, -.6) -- (9,-4.6) -- (4.6,-4.6) -- cycle;
        \node (D) at (6.75, -2.65) {\includegraphics[width=4cm]{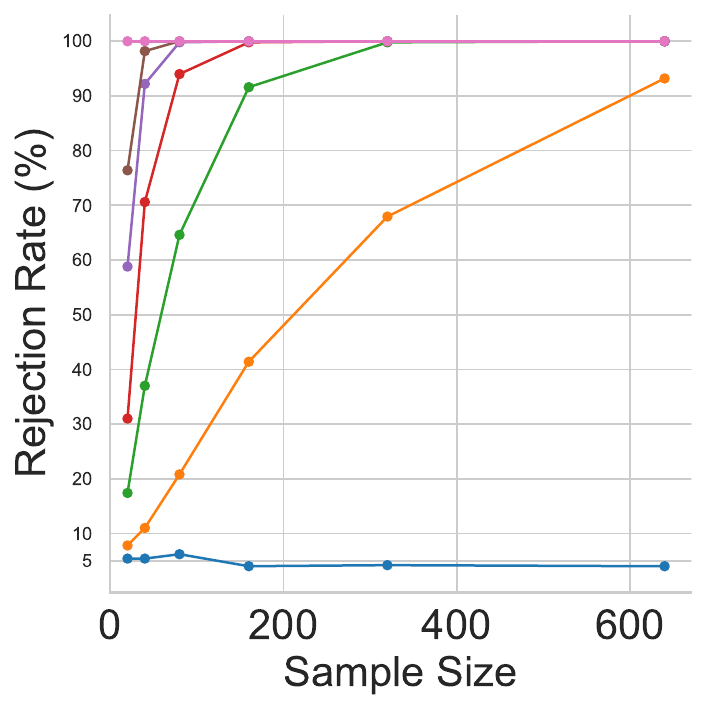}};
        \node[] (dd) at (4.8, -.8) {\small\textbf{d}};
        \draw (5, -.6) -- ++(0,-.4)  --  ++(-.4,0) {[rounded corners=4pt]-- ++(0,.4)} -- cycle;

        \draw[dotted, rounded corners=4pt] (-1.15, -4.7) -- ++(7, 0) -- ++(0,-.55) -- ++(-7,0) -- cycle;
        \node (L) at (2.35, -5.1) {\includegraphics[width=7cm]{plots/legend.pdf}};
        \node[] (ll) at (2.35, -4.9) {\tiny$\varphi$};
    \end{tikzpicture}   

    \caption{Panel (a) shows six consecutive densities sampled from the GAR(1) process described in Section \ref{subsec-wasserstein}. Panels (b), (c) and (d) are generated in the same way as in Figure \ref{fig:r}.}
    \label{fig:wasserstein}
\end{figure}

In this second experiment, we consider time series in the space $\c{D}([0, 1])$ of density functions over $[0, 1]$ equipped with the 2-Wasserstein distance, as described in Example \ref{ex-wasserstein}. Since the support of the distributions is bounded, the space $\c{D}([0,1])$ is bounded as well. Geodesics are given by linear interpolation of the corresponding quantile function: given two distributions $\mathbb{P}, \mathbb{Q} \in \c{D}([0,1])$ with quantile functions $F^{-1}_\mathbb{P}, F^{-1}_\mathbb{Q}$, the quantile function of any point on the connection geodesic $\gamma_\mathbb{P}^\mathbb{Q}$ is $F^{-1}_{\gamma_\mathbb{P}^\mathbb{Q}(t)}(u) = (1-t)F^{-1}_\mathbb{P}(u) + tF^{-1}_\mathbb{Q}(u)$.

We generate the time series with the standard normal distribution $N(0,1)$ truncated to $[0,1]$ as the Fr\'echet mean. Then, the data is generated according to Equation (\ref{eq-iterated}). The noise sampling is based on sampling a random optimal transport $\eta : [0,1] \rightarrow [0,1]$ and applying it by quantile composition, which corresponds to computing the pushforward under $\eta$: given a distribution $\mathbb{P} \in \c{D}([0,1])$ with quantile function $F_{\mathbb{P}}^{-1}$, the noise map $\varepsilon$ is then given by $F_{\varepsilon(\mathbb{P})}^{-1} = \eta \circ F^{-1}_{\mathbb{P}}$. To generate the transport maps $\eta$, we follow the procedure described in \citet{panaretos_amplitude_2016}. First, a random integer frequency is uniformly sampled from $\eset{-4, \ldots, 4} \backslash \eset{0}$, then, the maps are given by $\eta(x) = x - \sin\left(\pi k x\right)/ \abs{\pi k}$. The random maps $\eta$ are smooth, strictly increasing and satisfy $\eta(0) = 0$ and $\eta(1) = 1$. By symmetry of the random parameter $k$, one can see that for any $x \in [0,1]$, we have $\expec{\eta(x)} = x$, and this property is inherited by the noise maps $\varepsilon$. While the noise maps have a Lipschitz constant of 2, meaning that the condition of Theorem \ref{thm-wu} is satisfied for $\varphi < 0.5$, we observe empirically that the estimators seem to still be consistent even for values of $\varphi \in [0.5, 1)$.


The results in Figure \ref{fig:wasserstein} match those observed in the previous experiment. The scaled error curves displayed in panel (b) confirm the convergence rate proved in Theorem \ref{thm-rate-mu}. Similarly to the previous experiment, the conditions of Theorem \ref{thm-wu} are not satisfied for $\varphi = 1$, and the estimator $\hat\mu_T$ fails to converge; to improve readability, we did not include the associated curve. Similarly to the previous experiment, the estimator $\hat\varphi_T$ seems to converge at a $\sqrt{T}$ rate. The error does not seem to improve for i.i.d.\,observations, but the monotonicity as a function of $\varphi$ observed in the previous experiment approximately holds as seen in panel (c). We also observe in panel (d) that the hypothesis test behaves as expected. The blue curve, corresponding to the null hypothesis, demonstrates that the empirical size of the test is correct for all sample sizes considered.
\begin{figure}[t!]
    \begin{tikzpicture}
        \draw[dotted, rounded corners=4pt] (-4.2, 3.2) -- (9, 3.2) -- (9,-.5) -- (-4.2,-.5) -- cycle;
        \node (A) at (2.25, 2.2) {\includegraphics[width=12.5cm]{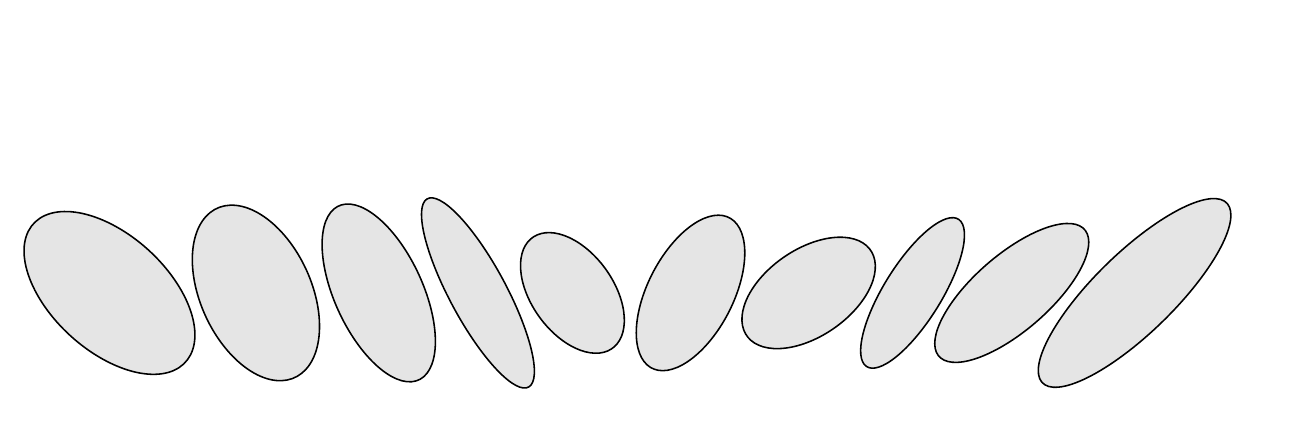}};
        \node[fill=white] (aa) at (-4, 3) {\small\textbf{a}};
        \draw (-3.8, 3.2) -- ++(0,-.4)  --  ++(-.4,0) {[rounded corners=4pt]-- ++(0,.4)} -- cycle;

        \node (x_start) at (-3.5, .05) {};
        \node (x_end) at (8, .05) {};
        \draw[->, line width=0.25mm, axis_color] (x_start) edge (x_end);
        \node (t) at (2.25, -.2) {Time};

        \draw[dotted, rounded corners=4pt] (-4.2, -.6) -- (0.1, -.6) -- (0.1,-4.6) -- (-4.2,-4.6) -- cycle;
        \node (B) at (-2.05, -2.65) {\includegraphics[width=4cm]{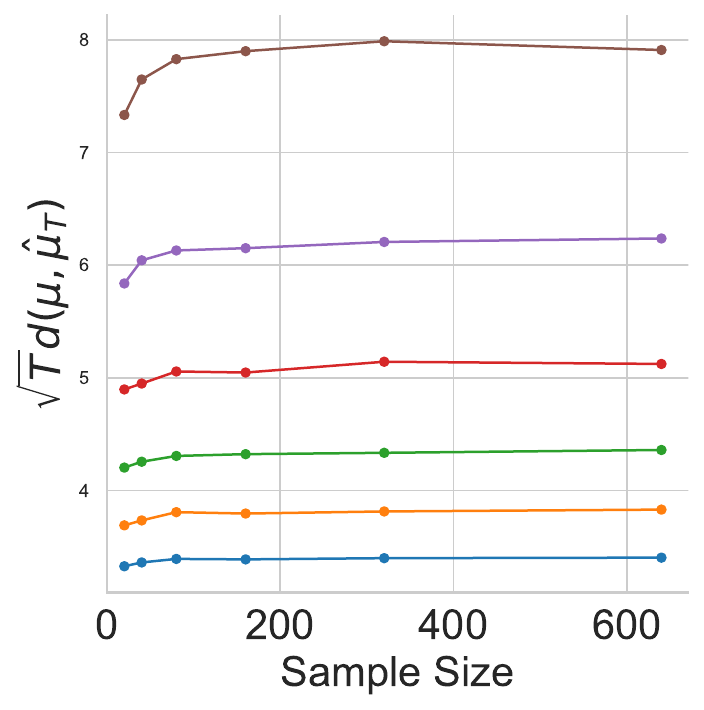}};
        \node[] (bb) at (-4, -.8) {\small\textbf{b}};
        \draw (-3.8, -.6) -- ++(0,-.4)  --  ++(-.4,0) {[rounded corners=4pt]-- ++(0,.4)} -- cycle;
        
        \draw[dotted, rounded corners=4pt] (.2, -.6) -- (4.5, -.6) -- (4.5,-4.6) -- (.2,-4.6) -- cycle;
        \node (C) at (2.35, -2.6) {\includegraphics[width=4cm]{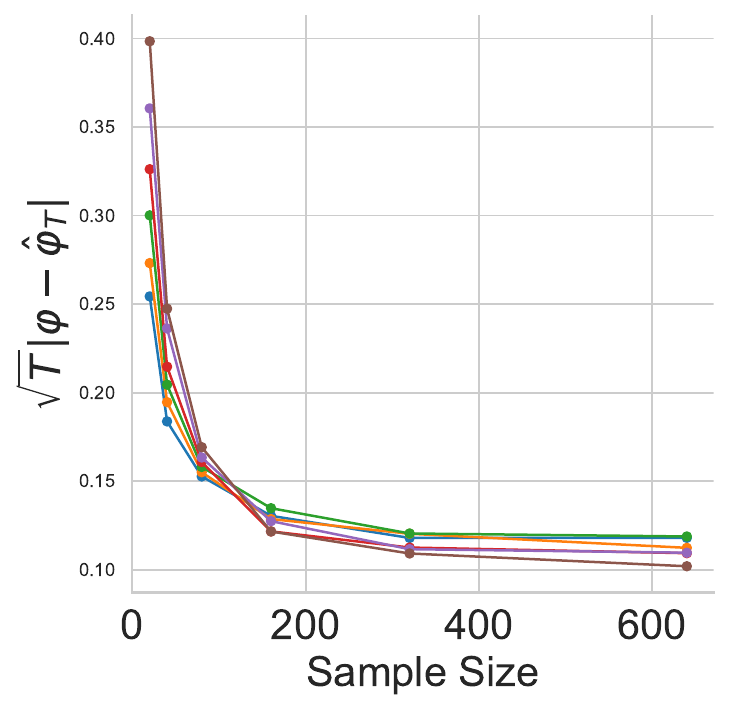}};
        \node[] (cc) at (.4, -.8) {\small\textbf{c}};
        \draw (.6, -.6) -- ++(0,-.4)  --  ++(-.4,0) {[rounded corners=4pt]-- ++(0,.4)} -- cycle;

        \draw[dotted, rounded corners=4pt] (4.6, -.6) -- (9, -.6) -- (9,-4.6) -- (4.6,-4.6) -- cycle;
        \node (D) at (6.75, -2.65) {\includegraphics[width=4cm]{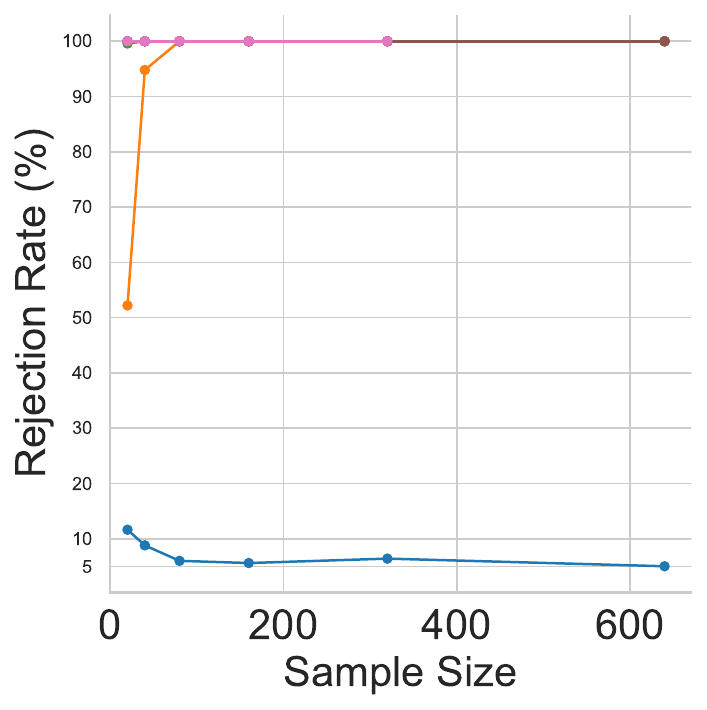}};
        \node[] (dd) at (4.8, -.8) {\small\textbf{d}};
        \draw (5, -.6) -- ++(0,-.4)  --  ++(-.4,0) {[rounded corners=4pt]-- ++(0,.4)} -- cycle;

        \draw[dotted, rounded corners=4pt] (-1.15, -4.7) -- ++(7, 0) -- ++(0,-.55) -- ++(-7,0) -- cycle;
        \node (L) at (2.35, -5.1) {\includegraphics[width=7cm]{plots/legend.pdf}};
        \node[] (ll) at (2.35, -4.9) {\tiny$\varphi$};
    \end{tikzpicture}   

    \caption{Panel (a) displays 10 consecutive covariance ellipses corresponding to the top-left $2 \times 2$ submatrix of the covariances sampled from the GAR(1) process described in Section \ref{subsec-wasserstein}. Each ellipse is the application of a covariance submatrix to a unit circle. Panels (b), (c) and (d) are generated in the same way as in Figure \ref{fig:r}.}
    \label{fig:logchol}
\end{figure}

\subsection{SPD matrices}

In this last experiment we investigate the properties of the GAR(1) model in the space $\mathcal{S}_{10}^+$ of $10\times 10$ SPD matrices with the Log-Cholesky distance described in Example \ref{ex-logchol}. In this space, matrices $M_0, M_1 \in \mathcal{S}_{10}^+$ are uniquely identified by their Cholesky factors $L_0, L_1$. Points on the geodesic line between these matrices are given by linearly interpolating off-diagonal entries of the Cholesky factors and geometrically interpolating the diagonal elements. That is, for $t \in [0,1]$, the Cholesky factor $L_t$ of $\gamma_{M_0}^{M_1}(t)$ is given via $\floor{L_t} = (1-t) \floor{L_0} + t\floor{L_t}$ and $D(L_t) = D(L_0)^{1-t}D(L_1)^t$.

We generate time series with the identity matrix $1_{10}$ as the Fr\'echet mean. Each noise map in this experiment applies a random congruent transformation of the input with a random lower-triangular matrix $L_\varepsilon \in \R^{10 \times 10}$ with $\varepsilon(X) = L_\varepsilon X L_\varepsilon^\top$. The lower-triangular entries of $L_\varepsilon$ are i.i.d. following a normal distribution $\floor{L_\varepsilon}_{ij} \sim N(0, 0.5^2)$, and the diagonal entries are i.i.d.\,following a log-normal distribution $\log D(L_\varepsilon)_{ii} \sim N(0, 0.2^2)$. For a matrix $X \in \mathcal{S}_{10}^+$ with Cholesky decomposition $X = LL^\top$, the matrix $\varepsilon(X)$ is also $\mathcal{S}_{10}^+$ and has Cholesky decomposition $\varepsilon(X) = L_\varepsilon L$, implying $\expec{\varepsilon(X)} = X$.

Figure \ref{fig:logchol} shows similar results as in the other two experimental settings. The convergence rate proved in Theorem \ref{thm-rate-mu} is confirmed in panel (b). In this setting, the stability of the error curves indicates an early attainment of the asymptotic regime. Panel (c) suggests a convergence rate of $\hat\varphi_T$ faster than $\sqrt{T}$, but additional simulations for larger sample sizes rejects this conjecture. We observe in panel (d) that the test exposes the correct level and high power already at small sample sizes for all non-zero tested values of $\varphi$.


\section{Application: Inflation expectation} \label{sec-application}


\begin{figure}[htbp]
    \begin{tikzpicture}
        \node[inner sep=0] (image1) at (0,0) {\includegraphics[width=0.45\textwidth]{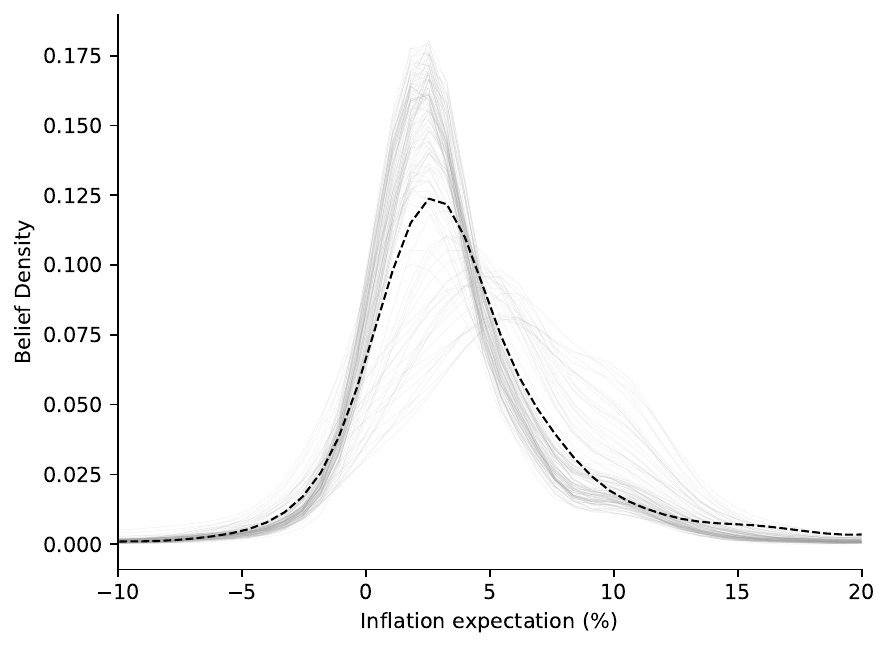}};
        \node[inner sep=0] (image2) at (0.48\textwidth,0.01\textheight) {\includegraphics[width=0.45\textwidth]{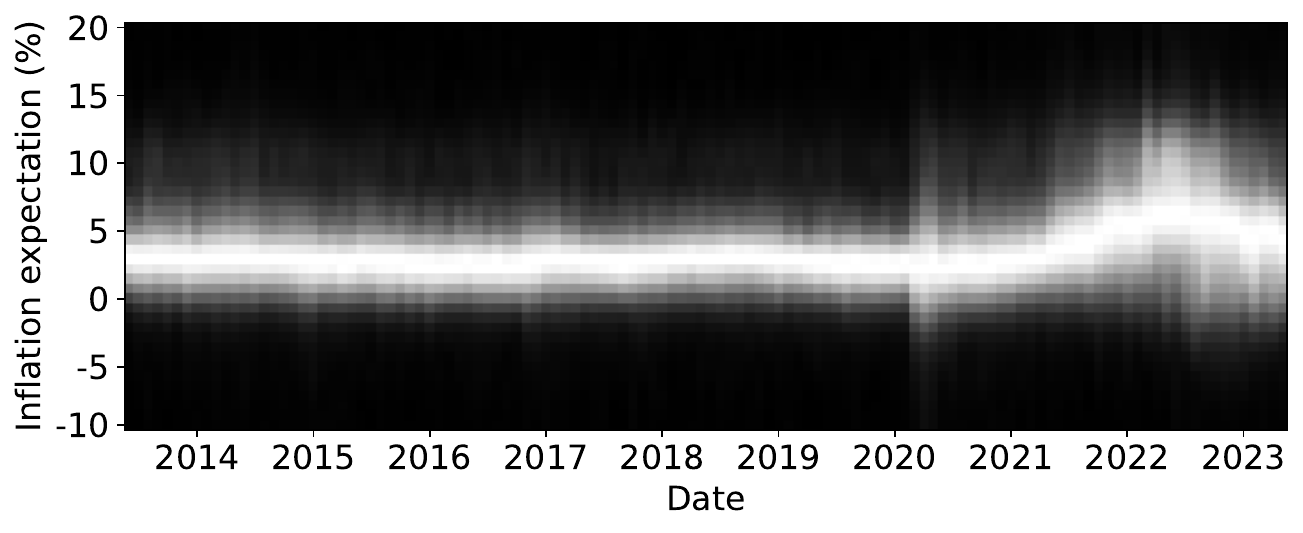}};
    \end{tikzpicture}
    \caption{Left: smoothed densities of the monthly 12-months-ahead inflation expectation. The empirical Fr\'echet mean of these densities is displayed in dashed black. Right: display of the monthly 12-months-ahead inflation expectation densities. In both panels, we only show the $[-10, 20]$ range of the data.}\label{fig-data-analysis-data}
\end{figure}

Analysis of consumer inflation expectations brings insights into how the general populations perceives broader economic trends \citep{dietrich_news_2022, meeks_heterogeneous_2023}. The \textit{Survey of Consumer Expectations} (SCE) is a monthly survey maintained by the Federal Reserve Bank of New York collecting information on households' expectations on a broad variety of economic topics between June 2013 and May 2023, see \citet{armantier_overview_2017}.

We focus our attention on the inflation expectation question, in which each consumer is asked to provide a distribution representing their belief for the 12-months ahead inflation. The survey respondents are presented with pre-defined bins over which they can distribute percentage points, defining a histogram of their beliefs. The bins are given by the nodes $-36\%$, $-12\%$, $-8\%$, $-4\%$, $-2\%$, $0\%$, $2\%$, $4\%$, $8\%$, $12\%$ and $36\%$. Each month, an average of approximately 1293 response histograms are available. We aggregate the histograms monthly by first taking the median belief of each histogram (which is already present in the dataset) and approximate the monthly median belief density via kernel density estimator with a Gaussian kernel and using Scott's rule \citep{scott_multivariate_1992} for the choice of the bandwidth. This results in a time-series of $T = 114$ elements in $\c{D}([-36, 36])$ displayed in Figure \ref{fig-data-analysis-data}. 

We fit the parameters of the GAR(1) model as described in Section \ref{sec-estimation} and obtain an empirical Fr\'echet mean $\hat\mu_T$ displayed in the left panel of Figure \ref{fig-data-analysis-data} (dashed black) and a concentration parameter $\hat\varphi_T = 0.85$ indicating a strong sequential dependence of the densities. The hypothesis test presented in Section \ref{sec-hypothesis-test} rejects the hypothesis of independence at level at a $5\%$ level with a test statistic $D_T \approx 0.76$ and estimated $p$-value $\hat p_B \approx 10^{-3}$ with $B = 1000$ permutations. The left panel of Figure \ref{fig-data-analysis-results} displays the histogram of the bootstrapped values of $D_T$, illustrating the approximate normal distribution of $D_T$ under $H_0$.

To further highlight the auto-regressive aspect of the data, we compare the residuals of the GAR(1) model with those obtained under the null model, $X_t = \varepsilon_t(\mu)$. Using fitted parameters, we generate predictions for each time step under the GAR(1) model as $\hat X_{t+1} = \gamma_{\hat\mu_T}^{X_t}(\hat\varphi_T)$, and under the null model as $\hat X_{t+1}^0 = \hat \mu_T$. We then consider the squared residuals $d(X_{t+1}, \hat X_{t+1})^2$ and $d(X_{t+1}, \hat \mu_T)^2$, respectively. To asses the model's performance, we compute a metric space adaptation of the coefficient of determination, as proposed by \citet{petersen_frechet_2019}. The empirical estimator $R^2_\oplus$ of $R^2_\oplus$ is given by
\begin{equation*}
    \hat R^2_\oplus = 1 - \frac{\sum_{t=1}^{T-1} d(X_{t+1}, \gamma_{\hat\mu_T}^{X_t}(\hat\varphi_T))^2}{\sum_{t=1}^{T-1} d(X_t, \hat\mu_T)^2}.
\end{equation*}

\begin{figure}[t]
    \begin{tikzpicture}
        \node[inner sep=0] (image1) at (0,0) {\includegraphics[width=0.45\textwidth]{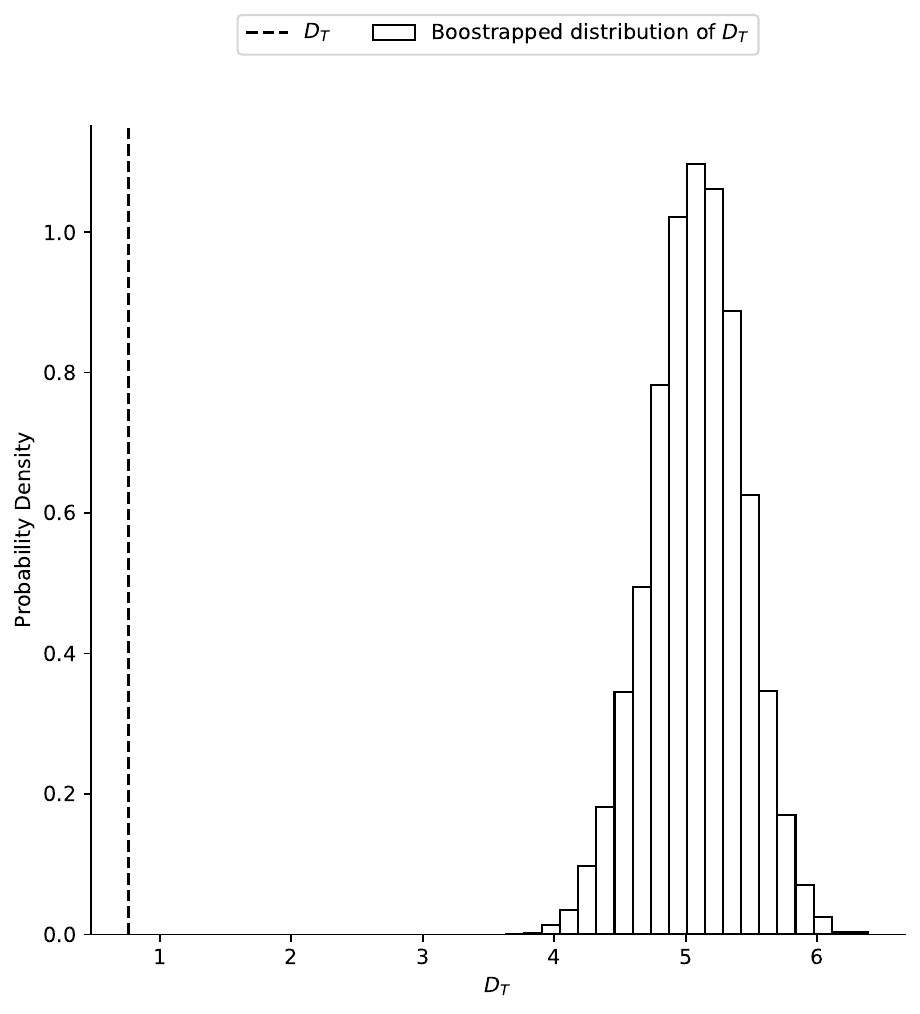}};
        \node[inner sep=0] (image3) at (0.5\textwidth,-0.34) {\includegraphics[width=0.45\textwidth]{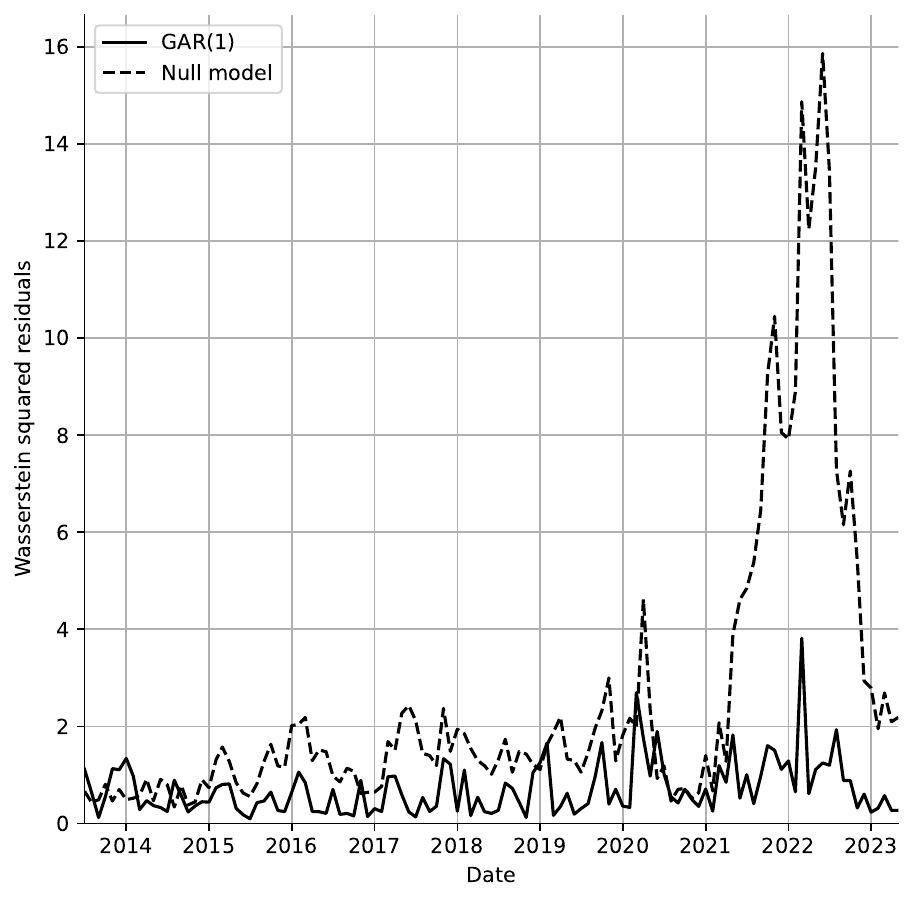}};
    \end{tikzpicture}
    \caption{Left: Histogram of the distribution of $D_T$ under the null hypothesis together with the computed value of $D_T$ on the observed data. Right, top: Empirical cumulative distribution function of the residuals for the GAR(1) model (full line) and null model (dashed line). Right, bottom: Residuals of both models over time.}\label{fig-data-analysis-results}
\end{figure}

In this analysis, we find an empirical coefficient of determination of $\hat R^2_\oplus = 0.72$ indicating that the GAR(1) model is able to explain a significant portion of the variability present in the data. The right panel of Figure \ref{fig-data-analysis-results} shows that the residuals of the GAR(1) fit are smaller than for the fit under the null model, as shown in the upper graph. Furthermore, as shown in the lower-right panel, the residuals under the null model residuals increase and are high during the years 2021 and 2022 while the GAR(1) residuals stay stable over time. This is consistent with the data as shown in the right panel of Figure \ref{fig-data-analysis-data} where it is visible that there is a shift in inflation expectation from 2021 onwards, possibly due to the economical impact of the COVID 19 pandemic and the escalation of the Russo-Ukranian war in early 2022. While this shift results in unexplained deviations from the mean in the null model starting in 2021, the time dependence of the GAR(1) model allows for a better fit in these years.
\section{Discussion} \label{sec-discussion}

This paper proposed a first-order autoregressive model for time series of random variable residing in metric spaces. The model is parametrized by a Fr\'echet mean and a concentration parameter which we proved can be consistently estimated under mild assumptions. This paper also presents a test for serial dependence under the GAR(1) model, allowing to test it against a null hypothesis of repeated i.i.d.\,measurements. Monte Carlo experiments as well as a real data analysis demonstrated the theoretical properties of the model as well as its practical relevance.

Several directions could be taken to extend the results of this paper. First, as observed in the experiments in Section \ref{sec-numerical}, a $\sqrt{T}$ rate of convergence of the estimator of the concentration parameter $\varphi$ appears to hold. This rate, as well as other stronger results about the asymptotic behavior of $\hat\varphi_T$, might be obtained using moment assumptions on the subderivates of $L$, see \citet{niemiro_asymptotics_1992}, or by assuming differentiability, see results from \citet{haberman_concavity_1989}. On the modeling side, two directions could be interesting to explore. A useful extension would be to adapt the model to allow for a negative relationship to the previous time step, with $\varphi < 0$. This can be done naturally in some specific cases by using an existing tangent space structure of the metric space, as done in \citet{Ghodrati_jtsa12736, zhu_autoregressive_2021}, but it is not necessarily clear how to define the notion of a \textit{negative direction} in a more general case. Furthermore, the model presented here only allows for a first-order auto-regressive structure. One possible extension would be to consider higher-order autoregressive models by applying the same principle, this time replacing the one geodesic update in Equation \eqref{eq-iterated} with multiple updates using previous time steps. However, we expect this approach to be challenging to analyze. Instead, it could be of interest to propose another auto-regressive model sharing the similarity to the AR(1) model on the real line, but which could be more easily extended to a higher number of lags. Finally, one could develop other classical tests found in time series analyses to this model class, for instance a test to detect change-points in the Fr\'echet mean or concentration parameter, see \citet{jiang_two-sample_2024}.

\section{Acknowledgement} \label{sec-ack}

This work has received funding from the European Union’s Horizon 2020 research and innovation program under the Marie Skłodowska-Curie grant agreement No 956107, "Economic Policy in Complex Environments (EPOC)".

\bibliographystyle{elsarticle-harv}
\bibliography{references}

\newpage

\appendix

\section{Comparison to other approaches} \label{sec-app-add-sims}

In this section, we compare the test presented in this paper to the serial independence test proposed in \citet{jiang_testing_2023} as well as a test for serial independence constructed by testing $\varphi = 0$ via the estimator $\hat\varphi_T$. We compare the rejection rate of the different tests at level $\alpha = 0.05$ for different simulation setups, values of $\varphi$ and sample sizes in the same three scenarios as in Section \ref{sec-numerical}: $\R$ with multiplicative noise, univariate distributions with the 2-Wasserstein distance and SPD matrices with the Log-Cholesky metric.

The first alternative test uses $\hat\varphi_T$ as the test statistic. Similarly to the test constructed via the statistic $D_T$, to construct a level $\alpha$ hypothesis test for $H_0 : \varphi = 0\,\t{vs.}\,H_1 : \varphi > 0$, we use a permutation procedure to compute approximate $p$-values under $H_0$. With the same notation as in Section \ref{sec-hypothesis-test}, let $\hat\varphi_T^{\pi}$ be the estimator of $\varphi$ computed on the randomly permuted sample $(X_{\pi(1)}, \ldots, X_{\pi(T)})$. The approximate $p$-value using a bootstrapped sample of $B \in \N$ replicas is $\hat p_B = \frac{1}{B}\sum_{b=1}^B \mathbbm{1}\eset{\hat\varphi_T \leq \hat\varphi_T^{\pi}}$. We reject the null hypothesis if $\hat p_B \leq \alpha$ to construct a level $\alpha$ test. While we do not have theoretical results about the asymptotic distribution of $\hat\varphi_T$ justifying this test, it is still of interest to empirically analyze the result of this procedure.

We also compare our test to the methodology presented in \citet{jiang_testing_2023}. The authors propose a generalization of the spectral density function to metric spaces based on the distance covariance, see \citet{lyons_distance_2013}. Based on this, the define two test statistics, $\textrm{CvM}_T$ and $\textrm{KS}_T$. Similarly to our $D_T$, their test statistics are non-pivotal and a wild bootstrap is proposed to obtain the critical values of their test statistics.

The summary of the results can be found in Table \ref{table-rej-rate}. We observe that the tests based on $D_T$ or $\hat\varphi_T$ seem to almost uniformly outperform the tests proposed by \citet{jiang_testing_2023}. This could partially be explained by the fact that our tests are tailored for the data generating process considered, while the tests in \citet{jiang_testing_2023} do not assume any structure on the autoregressive data generating process. We observe that the tests based on $\hat\varphi_T$ and $D_T$ expose approximately similar performance in the first scenario of real numbers with multiplicative noise. For the other two scenarios, the two tests seem to be equally well calibrated, up to stochastic difference, while the test based on $D_T$ achieves considerably higher power compared uniformly over every alternative and sample size we considered.

\begin{table}
    \small
    \centering
    \begin{tabular}{lll|p{1.35cm}|p{1.35cm}|p{1.35cm}|p{1.35cm}}
        \multicolumn{3}{c|}{} & \makecell{\centering $\textrm{CvM}_T$}  & \makecell{\centering $\textrm{KS}_T$}   & \makecell{\centering $\hat\varphi_T$} & \makecell{\centering $D_T$} \\
        \hhline{=======}
        \multirow{9}{*}{$\mathbb{R}$} 
        &    & $n=40$  & \makecell{\centering 0.112} & \makecell{\centering 0.090} & \makecell{\centering 0.058} & \makecell{\centering $\mathbf{\centering 0.050}$} \\ \cline{3-7}
        & $H_0: \varphi=0$                               & $n=80$  & \makecell{\centering 0.076} & \makecell{\centering 0.076} & \makecell{\centering $\mathbf{\centering 0.052}$} & \makecell{\centering 0.034} \\ \cline{3-7}
        &                                & $n=160$ & \makecell{\centering 0.090} & \makecell{\centering 0.074} & \makecell{\centering 0.044} & \makecell{\centering $\mathbf{\centering 0.048}$} \\
        \cline{2-7}
        & \multirow{3}{*}{$\varphi=0.1$} & $n=40$  & \makecell{\centering 0.120} & \makecell{\centering 0.106} & \makecell{\centering $\mathbf{\centering 0.166}$} & \makecell{\centering 0.142} \\ \cline{3-7}
        &                                & $n=80$  & \makecell{\centering 0.160} & \makecell{\centering 0.106} & \makecell{\centering 0.228} & \makecell{\centering $\mathbf{\centering 0.236}$} \\ \cline{3-7}
        &                                & $n=160$ & \makecell{\centering 0.144} & \makecell{\centering 0.066} & \makecell{\centering $\mathbf{\centering 0.312}$} & \makecell{\centering $\mathbf{\centering 0.312}$} \\
        \cline{2-7}
        & \multirow{3}{*}{$\varphi=0.3$} & $n=40$  & \makecell{\centering 0.238} & \makecell{\centering 0.108} & \makecell{\centering $\mathbf{\centering 0.492}$} & \makecell{\centering 0.478} \\ \cline{3-7}
        &                                & $n=80$  & \makecell{\centering 0.434} & \makecell{\centering 0.116} & \makecell{\centering 0.800} & \makecell{\centering $\mathbf{\centering 0.806}$} \\ \cline{3-7}
        &                                & $n=160$ & \makecell{\centering 0.664} & \makecell{\centering 0.118} & \makecell{\centering 0.936} & \makecell{\centering $\mathbf{\centering 0.972}$} \\
        \hhline{=======}
        \multirow{9}{*}{$\mathcal{D}$} 
        &    & $n=40$  & \makecell{\centering 0.156} & \makecell{\centering 0.092} & \makecell{\centering $\mathbf{\centering 0.056}$} & \makecell{\centering 0.036} \\ \cline{3-7}
        & $H_0: \varphi=0$                               & $n=80$  & \makecell{\centering 0.158} & \makecell{\centering 0.092} & \makecell{\centering $\mathbf{\centering 0.042}$} & \makecell{\centering 0.040} \\ \cline{3-7}
        &                                & $n=160$ & \makecell{\centering 0.146} & \makecell{\centering 0.114} & \makecell{\centering $\mathbf{\centering 0.048}$} & \makecell{\centering 0.044} \\
        \cline{2-7}
        & \multirow{3}{*}{$\varphi=0.1$} & $n=40$  & \makecell{\centering $\mathbf{\centering 0.206}$} & \makecell{\centering 0.114} & \makecell{\centering 0.132} & \makecell{\centering 0.180} \\ \cline{3-7}
        &                                & $n=80$  & \makecell{\centering 0.218} & \makecell{\centering 0.096} & \makecell{\centering 0.186} & \makecell{\centering $\mathbf{\centering 0.276}$} \\ \cline{3-7}
        &                                & $n=160$ & \makecell{\centering 0.248} & \makecell{\centering 0.098} & \makecell{\centering 0.150} & \makecell{\centering $\mathbf{\centering 0.530}$} \\
        \cline{2-7}
        & \multirow{3}{*}{$\varphi=0.3$} & $n=40$  & \makecell{\centering 0.336} & \makecell{\centering 0.138} & \makecell{\centering 0.596} & \makecell{\centering $\mathbf{\centering 0.798}$} \\ \cline{3-7}
        &                                & $n=80$  & \makecell{\centering 0.446} & \makecell{\centering 0.146} & \makecell{\centering 0.732} & \makecell{\centering $\mathbf{\centering 0.982}$} \\ \cline{3-7}
        &                                & $n=160$ & \makecell{\centering 0.616} & \makecell{\centering 0.122} & \makecell{\centering 0.818} & \makecell{\centering $\mathbf{1}$} \\
        \hhline{=======}
        \multirow{9}{*}{$\mathcal{S}_{10}^+$} 
        &    & $n=40$  & \makecell{\centering 0.850} & \makecell{\centering 0.344} & \makecell{\centering $\mathbf{\centering 0.012}$} & \makecell{\centering $\mathbf{\centering 0.012}$} \\ \cline{3-7}
        & $H_0: \varphi=0$                               & $n=80$  & \makecell{\centering 0.966} & \makecell{\centering 0.426} & \makecell{\centering $\mathbf{\centering 0.018}$} & \makecell{\centering 0.014} \\ \cline{3-7}
        &  & $n=160$ & \makecell{\centering 0.986} & \makecell{\centering 0.430} & \makecell{\centering 0.016} & \makecell{\centering $\mathbf{\centering 0.042}$} \\
        \cline{2-7}
        & \multirow{3}{*}{$\varphi=0.1$} & $n=40$  & \makecell{\centering 0.882} & \makecell{\centering 0.320} & \makecell{\centering 0.524} & \makecell{\centering $\mathbf{\centering 0.984}$} \\ \cline{3-7}
        &                                & $n=80$  & \makecell{\centering 0.978} & \makecell{\centering 0.368} & \makecell{\centering 0.618} & \makecell{\centering $\mathbf{1}$} \\ \cline{3-7}
        &                                & $n=160$ & \makecell{\centering 0.980} & \makecell{\centering 0.444} & \makecell{\centering 0.654} & \makecell{\centering \makecell{\centering $\mathbf{1}$}} \\
        \cline{2-7}
        & \multirow{3}{*}{$\varphi=0.3$} & $n=40$  & \makecell{\centering 0.824} & \makecell{\centering 0.332} & \makecell{\centering $\mathbf{1}$} & \makecell{\centering $\mathbf{1}$} \\ \cline{3-7}
        &                                & $n=80$  & \makecell{\centering 0.980} & \makecell{\centering 0.390} & \makecell{\centering $\mathbf{1}$} & \makecell{\centering $\mathbf{1}$} \\ \cline{3-7}
        &                                & $n=160$ & \makecell{\centering 0.992} & \makecell{\centering 0.396} & \makecell{\centering $\mathbf{1}$} & \makecell{\centering $\mathbf{1}$} \\
    \end{tabular}
    \caption{Empirical rejection rate for each test at level $\alpha = 0.05$. Each rejection rate is based on 1000 simulations, as described in the introduction to Section \ref{sec-numerical}. In each line, the number in bold corresponds to the rejection rate closest to the desired test level $\alpha = 0.05$ for the rows with $\varphi = 0$ and with the highest rejection rate for the rows with $\varphi > 0$.  } \label{table-rej-rate}
\end{table}

\newpage
\section{Proofs of theorems} \label{sec-proofs}
\subsection*{Consistency of the mean estimator}

We start by defining the following function which will be useful in the proofs presented in this section. Let $\omega, \omega_0 \in \Omega$, we define for all $x \in \Omega$ the function
\begin{equation*}
    g^\omega_{\omega_0}(x) = d(x, \omega)^2 - d(x, \omega_0)^2.
\end{equation*}
In a Hadamard space, $g^\omega_{\omega_0}$ has the following Lipschitz property holding both in $x$ and in the pair $(\omega, \omega_0)$.
\begin{blemma} \label{lem-g-lip}
    Let $(\Omega, d)$ be a Hadamard space and $\omega, \omega_0, x, x' \in \Omega$, then
    \begin{equation*}
        \abs{ g_{\omega_0}^\omega(x) - g_{\omega_0}^\omega(x') } \leq 2d(\omega, \omega_0) d(x, x').
    \end{equation*}
\end{blemma}
\begin{proof}
    By Reshetnyak’s Quadruple Comparison (see Proposition \ref{prop-quadruple-comp}),
    \begin{align*}
        & d(x', \omega_0)^2 + d(x, \omega)^2 \leq d(x, \omega_0)^2 + d(x', \omega)^2 + 2d(x, x')d(\omega, \omega_0)\\
        &\Rightarrow d(x, \omega)^2 - d(x', \omega)^2 - \left(d(x, \omega_0)^2 - d(x', \omega_0)^2\right) \leq 2d(x, x')d(\omega, \omega_0)\\
        &\Rightarrow g_{\omega_0}^\omega(x) - g_{\omega_0}^\omega(x') \leq 2d(\omega, \omega_0)d(x, x')
    \end{align*}
    By inverting the role of $x$ and $x'$ we obtain the same upper bound on $g_{\omega_0}^\omega(x') - g_{\omega_0}^\omega(x)$ which completes the proof.
\end{proof}

In order to study the asymptotic behavior of the minimizer $\hat\mu_T$ of $M_T$, we need to quantify the deviations of the empirical loss function $M_T$ from its population limit $M$. To that end, given some $\omega_0 \in \Omega$, we define the process $\omega \mapsto H_{\omega_0}^\omega$ by
\begin{equation}\label{def-h}
    H_{\omega_0}^\omega = \frac{1}{\sqrt{T}} \sum_{t=1}^T g_{\omega_0}^\omega(X_t) - \expec{g_{\omega_0}^\omega(X)}.
\end{equation}
We start with the following proposition showing that $H_{\omega_0}^\omega$ is sub-Gaussian.

\begin{bproposition}\label{prop-subgaussian}
    Under the assumptions of Theorem \ref{thm-rate-mu}, there exists constants $c_1, c_2 > 0 $ such that for all $\lambda > 0$, 
    \begin{equation}\label{eq-prop-subgaussian}
    \P{\abs{H^\omega_{\omega_0}} \geq \lambda}
    \leq
    c_1 \exp\left\{- \frac{\lambda^2}{c_2 d(\omega,\omega_0)^2} \right\}.
    \end{equation}
\end{bproposition}

\begin{proof}
Following \citet{wu_limit_2004} and \citet{gordin_central_1978}, we study the asymptotic behavior of the scaled process $\sqrt{T}H^\omega_{\omega_0}$ by considering the solution $h \in L_2(\Omega)$ to Poisson's equation 
\begin{equation}\label{eq-poisson}
    h(x) - \expec{h(X_{t+1}) \mid X_t = x } = g_{\omega_0}^\omega(x) - \expec{g^\omega_{\omega_0}(X)}.
\end{equation}
A solution to this equation exists and can be written as
\begin{equation*}
    h(x) = \sum_{t=0}^\infty \left(\expec{g^\omega_{\omega_0}(X_t) \mid X_0 = x}  - \expec{g^\omega_{\omega_0}(X)}\right).
\end{equation*}
Using (\ref{eq-poisson}), we can decompose $\sqrt{T}H^\omega_{\omega_0}$ as
\begin{align*}
    \sqrt{T}H^\omega_{\omega_0} 
    &= \sum_{t=1}^T g_{\omega_0}^\omega(X_t) - \expec{g_{\omega_0}^\omega(X)}
    = \sum_{t=1}^T h(X_t) - \expec{h(X_{t+1}) \mid X_t}\\
    &= \expec{h(X_1) \mid X_{0}} - \expec{h(X_{T+1}) \mid X_{T}} + \sum_{t=1}^T h(X_t) - \expec{h(X_t) \mid X_{t-1}}\\
    &:= R_T + \sum_{t=1}^T D_t,
\end{align*}
where we introduced $R_T = \expec{h(X_1) \mid X_{0}} - \expec{h(X_{T+1}) \mid X_{T}}$ and $D_t = h(X_t) - \expec{h(X_t) \mid X_{t-1}}$.
Note that $D_t$ is a martingale-difference and that $R_T = O_P(1)$. To show that this decomposition is valid, we start by showing that $h$ is absolutely summable. That is, we show
\begin{equation*}
    \sum_{t=0}^\infty \abs{\expec{ g^\omega_{\omega_0}(X_t) \mid X_0 = x } - \expec{g^\omega_{\omega_0}(X)}} < \infty.
\end{equation*}
Using Lemma \ref{lem-g-lip} and the independence of $\eset{\varepsilon_t}_{t \in \N}$, along with the assumption that condition (\ref{eq-geom-cond}) holds for some $\alpha > 1$ (and hence also for $\alpha = 1$), we have
\begin{align*}
    &\abs{\expec{ g^\omega_{\omega_0}(X_t) \mid X_0 = x } - \expec{g^\omega_{\omega_0}(X)}}
    \leq 2 d(\omega, \omega_0) \expect{X_t}{ \expect{X}{d(X_t, X)} \mid X_0 = x }\\
    &\qquad= 2 d(\omega, \omega_0) \expec{ d(X_t(x), X_t(\tilde X_0)) }
    \leq 2 d(\omega, \omega_0) C r^t.
\end{align*}
Using this bound in the infinite sum, we obtain
\begin{equation*}
    \sum_{t=0}^\infty \abs{\expec{ g^\omega_{\omega_0}(X_t) \mid X_0 = x } - \expec{g^\omega_{\omega_0}(X)}} \leq 2 d(\omega, \omega_0) C \sum_{t=0}^\infty r^t = \tilde{C}d(\omega, \omega_0).
\end{equation*}
Hence, $\lim_{T \rightarrow \infty} \sqrt{T}H^\omega_{\omega_0}$ exists and is almost surely bounded implying that the decomposition presented above is valid and that both $D_t$ and $R_t$ are also almost surely bounded with $\abs{D_t} < C$ and $\abs{R_T} < C$. In particular, we have that $\abs{D_t} \leq \tilde{C} d(\omega,\omega_0)$ for some $\tilde{C} > 0$, since
\begin{align*}
    D_t 
    &= h(X_t) - \expec{h(X_t) \mid X_{t-1}}\\
    &= \sum_{k=0}^\infty \expec{g^\omega_{\omega_0}(X_k) \mid X_t} - \expec{\expec{g^\omega_{\omega_0}(X_k) \mid X_t} \mid X_{t-1}}\\
    &= \sum_{k=0}^\infty \expec{g^\omega_{\omega_0}(X_k) \mid X_t} - \expec{g^\omega_{\omega_0}(X_k)\mid X_{t-1}}\qquad\t{(Tower rule)}\\
    &= \sum_{k=t}^\infty \expec{g^\omega_{\omega_0}(X_k) \mid X_t} - \expec{g^\omega_{\omega_0}(X_k)\mid X_{t-1}}.\qquad(k \leq t-1 \Rightarrow \sigma(X_k) \subset \sigma(X_{t-1}))
\end{align*}
We can use the function notation in Equation (\ref{eq-f-def}) to get $X_k = X_{k:t+1}(X_t)$ and rewrite
\begin{align*}
    &\expec{g^\omega_{\omega_0}(X_k) \mid X_t} - \expec{g^\omega_{\omega_0}(X_k)\mid X_{t-1}}
    \\&= \expec{g^\omega_{\omega_0}(X_{k:t+1}(X_t)) \mid X_t} - \expec{g^\omega_{\omega_0}(X_{k:t+1}(X_t))\mid X_{t-1}}
\end{align*}
where $X_{k:t+1}$ is random in both conditional expectations, but $X_t$ is only random in the second conditional expectation. Taking absolute values and using Lemma \ref{lem-g-lip} together with (2) in \citet{wu_limit_2004} gives
\begin{equation*}
    \abs{\expec{g^\omega_{\omega_0}(X_k) \mid X_t} - \expec{g^\omega_{\omega_0}(X_k)\mid X_{t-1}}} \leq 2 d(\omega, \omega_0) C r^{k - (t+1)}.
\end{equation*}
Using this bound in the sum gives
\begin{equation*}
    \abs{D_t} \leq \sum_{k=t}^\infty 2 d(\omega, \omega_0) C r^{k - (t+1)} = 2C d(\omega, \omega_0) \sum_{k=0}^\infty r^k = \tilde{C} d(\omega, \omega_0).
\end{equation*}
We now show the following result, which is equivalent to \eqref{eq-prop-subgaussian}
\begin{equation}\label{eq-sn-bound}
    \P{\abs{R_T + \sum_{t=1}^T D_t} \geq T\lambda}
    \leq
    e^2 \exp\left\{- \frac{T \lambda^2}{4 C d(\omega,\omega_0)^2} \right\}.
\end{equation}
To do this, we consider two cases. If $\sqrt{T}\lambda < 4C$, the bound is vacuous. Otherwise, for $\sqrt{T}\lambda \geq 4C$, we have
\begin{equation*}
    \P{\abs{R_T + \sum_{t=1}^T D_t} \geq T\lambda}
    \leq \P{\abs{R_T} \geq \frac{T\lambda}{4}} + \P{\abs{\sum_{t=1}^T D_t} \geq \frac{3T\lambda}{4}}.
\end{equation*}
Since $\sqrt{T}\lambda \geq 4C$ and $\abs{R_t} \leq C$, we have that $\eset{\abs{R_T} \geq \frac{T\lambda}{4}}$ is a probability zero event. We can thus focus on bounding the martingal difference sum. By Chernoff's bounding technique, we have for all $\lambda > 0$
\begin{align*}
    \P{\sum_{t=1}^T D_t \geq \frac{3T\lambda}{4}}
    &= \P{\frac{1}{\sqrt{T}}\sum_{t=1}^T D_t \geq \frac{3\sqrt{T}\lambda}{4}}\\
    &\leq \min_{u > 0} \exp\left(-u\frac{3\sqrt{T}\lambda}{4}\right)\expec{\exp\left(u \frac{1}{\sqrt{T}}\sum_{t=1}^T D_t \right)}.
\end{align*}
Using the bound $\abs{D_t} \leq \tilde{C} d(\omega, \omega_0)$ and following the proof of Azuma-Hoeffding's inequality (see Theorem 2.2.1 in \citet{raginsky_concentration_2013}) we can bound the martingale-difference sum,
\begin{align*}
    \P{\sum_{t=1}^T D_t \geq \frac{3T\lambda}{4}}
    &\leq \min_{u > 0} \exp\left(-u\frac{3\sqrt{T}\lambda}{4} + \frac{u^2}{2}\tilde{C}^2 d(\omega,\omega_0)^2 \right)
    \\&\leq \exp\left(-\frac{\lambda^2 T}{4\tilde{C}^2 d(\omega,\omega_0)^2} \right).
\end{align*}
Where the last inequality comes from taking $u = \lambda \sqrt{T} / (\tilde{C}^2 d(\omega,\omega_0)^2)$. By symmetry the arguments can be repeated on the mirrored inequality to obtain a bound on $\P{\abs{\sum_{t=1}^T D_t} \geq \frac{3T\lambda}{4}}$. Together with the previous argument, this gives
\begin{equation*}
    \P{\abs{R_T + \sum_{t=1}^T D_t} \geq T\lambda}
    \leq e^2 \exp\left(-\frac{\lambda^2 T}{4\tilde{C}^2 d(\omega,\omega_0)^2} \right).
\end{equation*}

\end{proof}

Using the fact that the empirical process applied to $g^\omega_{\omega_0}$ is sub-Gaussian, we can use standard M-estimation theory to provide a proof of Theorem \ref{thm-rate-mu}.

\begin{proof}[Proof of Theorem \ref{thm-rate-mu}]
Noting that $(M_T - M)(\omega) - (M_T - M)(\mu) = T^{-1/2}H_\mu^\omega$ we have by Proposition \ref{prop-subgaussian} that
\begin{equation*}
    \P{\sqrt{T}\abs{(M_T - M)(\omega) - (M_T - M)(\mu)} \geq \lambda}
    \leq
    c_1 \exp\left\{- \frac{\lambda^2}{c_2 d(\omega,\omega_0)^2} \right\}.
\end{equation*}
So $\eset{\sqrt{T}(M_T - M)(\omega)}_{\omega \in \Omega}$ is sub-Gaussian. By Corollary 2.2.8 in \citet{van_der_vaart_weak_1996}, we have
\begin{align*}
    \expec{\sup_{d(\omega, \mu) < \delta} \sqrt{T}\abs{(M_T - M)(\omega) - (M_T - M)(\mu)}}
    &\lesssim \int_0^\delta \sqrt{\log\left(1 + D(\varepsilon, d) \right)}\d \varepsilon\\
    &= \delta \int_0^1 \sqrt{\log\left(1 + D(\delta\varepsilon, d) \right)}\d \varepsilon.
\end{align*}
Since by assumption the entropy integral is bounded and $O(1)$ for $\delta \rightarrow 0$, we bound (up to a multiplicative constant) the modulus of continuity by $T^{-1/2}\delta$. Additionally, by the variance inequality in Hadamard spaces (see Proposition 4.4 in \citet{auscher_probability_2003}), we have that the condition $M(\omega) - M(\mu^\star) \geq d(\omega, \mu^\star)^2$ holds. Thus by Theorem 3.2.5 in \citet{van_der_vaart_weak_1996}, $d(\hat\mu_T, \mu^\star) = O_P(T^{-1/2})$.
\end{proof}

\subsection*{Uniform convergence of $L_T$}

\begin{bproposition}\label{prop-uniform}
    Under the conditions of Theorem \ref{thm-phi-consistent}, we have that \linebreak $\norm{L_T - L}_\infty = o_P(1)$.
\end{bproposition}

\begin{proof}
    We show this result by verifying the conditions of Corollary 2.2 of \citet{newey_uniform_1991}. Namely, we need to show that:
    \begin{enumerate}
        \item $L$ is continuous;
        \item $L_T$ converges pointwise to $L$;
        \item There exists a sequence $C_t = O_P(1)$ such that for all $\varphi, \varphi' \in (0,1)$, $\abs{L_T(\varphi) - L_T(\varphi')} \leq C_T \abs{\varphi - \varphi'}$.
    \end{enumerate} 
    We proceed to verify these conditions.

    \textit{1. Continuity of $L$.} By definition $\omega \mapsto d(\omega_0, \omega)^2$ is continous. Since $\Omega$ is a Hadamard space, we also have that geodesics are continuous in $t \in [0,1]$, hence for all $x_t, x_{t+1} \in \Omega$ and $\varphi_0 \in [0,1]$, $d(x_{t+1}, \gamma_{\mu}^{x_t} \varphi_n) \rightarrow d(x_{t+1}, \gamma_{\mu}^{x_t} \varphi_0)$ for any sequence $\varphi_n \rightarrow \varphi_0$. Furthermore, by geodesic convexity of the squared distance, we have that $d(x_{t+1}, \gamma_{\mu}^{x_t} \varphi)^2 \leq (1-\varphi) d(x_{t+1}, \mu)^2 + \varphi d(x_{t+1}, x_t)^2$ which is integrable with respect to $(X_t, X_{t+1})$ since $X_t$ and $X_{t+1}$ have second moments. By dominated convergence, this shows that $L(\varphi) \rightarrow L(\varphi_0)$ as $\varphi_n \rightarrow \varphi_0$ for a ny sequence $\varphi_n \rightarrow \varphi_0$, and hence $L$ is continuous.


    \textit{2. Pointwise convergence}. Let $\varphi_0 \in (0, 1)$. Using the fact that $X_{t+1} = \varepsilon_{t+1}(\gamma_\mu^{X_t}(\varphi))$, we can decompose the pointwise deviation of $L_T$ from $L$ as follows,
    \begin{align*}
        &\abs{L_T(\varphi_0) - L(\varphi_0)}\\
        &= \abs{\frac{1}{T-1}\sum_{t=1}^{T-1} d\left(X_{t+1}, \gamma_{\hat\mu_T}^{X_t}(\varphi_0)\right)^2 - \expec{d\left(\varepsilon_{t+1}(\gamma_{\mu}^{X_t}(\varphi)), \gamma_{\mu}^{X_t}(\varphi_0)\right)^2}}\\
        &\leq
        \abs{\frac{1}{T-1}\sum_{t=1}^{T-1} d\left(X_{t+1}, \gamma_{\hat\mu_T}^{X_t}(\varphi_0)\right)^2 - d\left(X_{t+1}, \gamma_{\mu}^{X_t}(\varphi_0)\right)^2}\\
        &\qquad+
        \abs{\frac{1}{T-1}\sum_{t=1}^{T-1} d\left(X_{t+1}, \gamma_{\mu}^{X_t}(\varphi_0)\right)^2 - \expec{d\left(\varepsilon_{t+1}(\gamma_{\mu}^{X_t}(\varphi)), \gamma_{\mu}^{X_t}(\varphi_0)\right)^2}}.
    \end{align*}
    By Lipschitz continuity of the squared distance in a bounded metric space, together with its geodesic convexity in Hadamard spaces and the fact that $d(\mu, \hat\mu_T) = O_P(T^{-1/2})$, we have that the first sum in the upper bound is $O_P(T^{-1/2})$.

    \begin{align*}
        &\abs{d\left(X_{t+1}, \gamma_{\hat\mu_T}^{X_t}(\varphi_0)\right)^2 - d\left(X_{t+1}, \gamma_{\mu}^{X_t}(\varphi_0)\right)^2}\\
        &\leq C_1 d\left(\gamma_{\hat\mu_T}^{X_t}(\varphi_0), \gamma_{\mu}^{X_t}(\varphi_0)\right)\abs{d\left(X_{t+1}, \gamma_{\hat\mu_T}^{X_t}(\varphi_0)\right) + d\left(X_{t+1}, \gamma_{\mu}^{X_t}(\varphi_0)\right)}.
    \end{align*}
    By the geodesic comparison inequality, $d\left(\gamma_{\hat\mu_T}^{X_t}(\varphi_0), \gamma_{\mu}^{X_t}(\varphi_0)\right) \leq \varphi_0d(\mu, \hat\mu_T)$, and using that $x \mapsto d(x_0, x)$ is geodesically convex, we get
    \begin{align*}
        d(X_{t+1}, \gamma_{\hat\mu_T}^{X_t}(\varphi_0))
        &\leq \varphi_0 d(X_{t+1}, \hat\mu_T) + (1-\varphi_0)d(X_{t+1}, X_t)\\
        &\leq d(X_{t+1}, \hat\mu_T) + d(X_{t+1}, X_t)\\
        &\leq 2 d(X_{t+1}, \mu) + d(X_t, \mu) + d(\mu, \hat\mu_T).
    \end{align*}
    Similarly, $d\left(X_{t+1}, \gamma_{\mu}^{X_t}(\varphi_0)\right) \leq 2d(X_{t+1}, \mu) + d(X_t, \mu)$, giving
    \begin{align*}
        &\abs{d\left(X_{t+1}, \gamma_{\hat\mu_T}^{X_t}(\varphi_0)\right)^2 - d\left(X_{t+1}, \gamma_{\mu}^{X_t}(\varphi_0)\right)^2}\\
        &\leq C_2 d(\mu, \hat\mu_T)\left[d(X_{t+1}, \mu) + d(X_t, \mu) + d(\mu, \hat\mu_T)\right]
    \end{align*}
    Taking the average over $t = 1, \ldots, T-1$, we get
    \begin{align*}
        &\abs{\frac{1}{T-1}\sum_{t=1}^{T-1} d\left(X_{t+1}, \gamma_{\hat\mu_T}^{X_t}(\varphi_0)\right)^2 - d\left(X_{t+1}, \gamma_{\mu}^{X_t}(\varphi_0)\right)^2}\\
        &\dot\leq d(\mu, \hat\mu_T)^2 + d(\mu, \hat\mu_T) \frac{1}{T-1} \sum_{t=1}^{T-1} d(X_{t+1}, \mu) + d(X_t, \mu).
    \end{align*}
    We now show that the second term is $O_P(T^{-1/2})$ as well. We do this using Theorem 3 in \citet{wu_limit_2004} with $Y_t = (X_t, X_{t+1})$ and $g(X_t, X_{t+1}) = d\left(X_{t+1}, \gamma_{\mu}^{X_t}(\varphi_0)\right)^2 - \expec{d\left(\varepsilon_{t+1}(\gamma_{\mu}^{X_t}(\varphi)), \gamma_{\mu}^{X_t}(\varphi_0)\right)^2}$. Let $\rho$ be the product metric on $\Omega \times \Omega$, $\rho((x_1, x_2), (y_1, y_2)) = \sqrt{d(x_1, y_1)^2 + d(x_2, y_2)^2}$. Let $Y_t = (X_t, X_{t+1})$ and $\tilde Y_t = (\tilde X_t, \tilde X_{t+1})$ be pairs in $\Omega \times \Omega$ such that $\rho(Y_t, \tilde Y_t) \leq \delta$, then
    \begin{align*}
        &\abs{g(Y_t) - g(\tilde Y_t)}
        = \abs{d\left(X_{t+1}, \gamma_{\mu}^{X_t}(\varphi_0)\right)^2 - d\left(\tilde X_{t+1}, \gamma_{\mu}^{\tilde X_t}(\varphi_0)\right)^2}\\
        &\leq \abs{d\left(X_{t+1}, \gamma_{\mu}^{X_t}(\varphi_0)\right)^2 - d\left(\tilde X_{t+1}, \gamma_{\mu}^{X_t}(\varphi_0)\right)^2}\\
        &\qquad\qquad+ \abs{d\left(\tilde X_{t+1}, \gamma_{\mu}^{X_t}(\varphi_0)\right)^2 - d\left(\tilde X_{t+1}, \gamma_{\mu}^{\tilde X_t}(\varphi_0)\right)^2}\\
        &\leq C d(X_{t+1}, \tilde X_{t+1}) + C d\left(\gamma_{\mu}^{X_t}(\varphi_0), \gamma_{\mu}^{\tilde X_t}(\varphi_0)\right)\\
        &= C d(X_{t+1}, \tilde X_{t+1}) + C\varphi_0 d(X_t, \tilde X_t)
    \end{align*}
    Since $\rho(Y_t, \tilde Y_t) \leq \delta$, we have that $\max \eset{d(X_{t+1}, \tilde X_{t+1}), d(X_t, \tilde X_t)} \leq \delta$ and hence
    $\abs{g(Y_t) - g(\tilde Y_t)} \leq C \delta$, showing that $g$ is Dini continuous and also stochastically Dini continuous. Theorem 3 in \citet{wu_limit_2004} gives that the second term in the above equation converges to a Brownian motion when scaled by $\sqrt{T}$ and hence is $O_P(T^{-1/2})$ which completes the proof of pointwise convergence.

    \textit{3. Stochastic Lipschitz Continuity of $L_T$.} Let $\varphi, \varphi' \in (0,1)$, then using that $\Omega$ is bounded and thus the squared distance is Lipschitz, we have that
    \begin{align*}
        \abs{L_T(\varphi) - L_T(\varphi')}
        &\leq \frac{1}{T-1}\sum_{t=1}^{T-1} \abs{d\left(X_{t+1}, \gamma_{\hat\mu_T}^{X_t}(\varphi)\right)^2 - d\left(X_{t+1}, \gamma_{\hat\mu_T}^{X_t}(\varphi')\right)^2}\\
        &\leq C \frac{1}{T-1}\sum_{t=1}^{T-1} d(\gamma_{\hat\mu_T}^{X_t}(\varphi), \gamma_{\hat\mu_T}^{X_t}(\varphi'))\\
        &= \abs{\varphi - \varphi'} C \frac{1}{T-1}\sum_{t=1}^{T-1} d(\hat\mu_T, X_t).
    \end{align*}
    Again using that $\Omega$ is bounded, the average is also bounded and we obtain the desired result.
\end{proof}

\section{Auxiliary theoretical results} \label{sec-app-hadamard}

\subsection*{General results in Hadamard spaces}

We start by stating results available in Hadamard spaces that will be used in the rest of the Appendix. 

\begin{cproposition}[Reshetnyak's Quadruple Comparison; Proposition 2.4 in \citet{auscher_probability_2003}]\label{prop-quadruple-comp}
    
    Let $(\Omega, d)$ be a Hadamard space. For all $x_1, x_2, x_3, x_4 \in \Omega$,
    \begin{equation*}
        d(x_1, x_3)^2+d(x_2, x_4)^2 \leq d(x_2, x_3)^2+d(x_4, x_1)^2+2 d(x_1, x_2) d(x_3, x_4).
    \end{equation*}
\end{cproposition}

Specializing this inequality to geodesics yields the following 

\begin{cproposition}[Geodesic Comparison Inequality; Corollary 2.5 in \citet{auscher_probability_2003}]\label{prop-geo-comp}
    Let $(\Omega, d)$ be a Hadamard space, $\gamma, \eta : [0, 1] \rightarrow \Omega$ be geodesics and $t \in [0,1]$. Then
    \begin{align*}
        d(\gamma(t), \eta(t))^2 \leq &(1-t)d(\gamma(0), \eta(0))^2 + td(\gamma(1), \eta(1))^2 \\
            &\qquad- t(1-t)[d(\gamma(0), \gamma(1) - d(\eta(0), \eta(1))]^2.
    \end{align*}
\end{cproposition}

\subsection*{Strong convexity of $L$}

It is possible to extend the identifiability result in Theorem \ref{thm-phi-id} and show that $L$ does not only have a unique minimizer, but is also strongly convex. This, as we show in the following lemma is a consequence of the geodesic convexity of the squared distance in Hadamard spaces.

\begin{clemma}\label{lem-L-strongly-convex}
    Let $\Xts \subset \Omega$, assume that $\Xts$ are $L^2(\Omega)$ and satisfies Equation (\ref{eq-iterated}) with true concentration parameter $\varphi \in [0,1]$. Then, the function $L$ is strongly convex.
  \end{clemma}
  \begin{proof}
    We show that the strong convexity of $L$ is inherited from the geodesic convexity of the squared distance in Hadamard spaces. Indeed, let $\varphi_1, \varphi_2 \in [0,1]$, wlog $\varphi_1 < \varphi_2$. Let $t \in [0,1]$ and define $\varphi_t = (1-t)\varphi_1 + t\varphi_2$. Then, $\gamma_{\mu}^{X_t}$ restricted to $[\varphi_1, \varphi_2]$ and reparametrized on $[0,1]$ gives the geodesic connecting $\gamma_{\mu}^{X_t}(\varphi_1)$ to $\gamma_{\mu}^{X_t}(\varphi_2)$ and hence
    \begin{equation*}
        L(\varphi_t)
        = \expec{d\left(X_{t+1}, \gamma_{\mu}^{X_t}((1-t)\varphi_1 + t\varphi_2)\right)^2}
        = \expec{d\left(X_{t+1}, \gamma_{\gamma_{\mu}^{X_t}(\varphi_1)}^{\gamma_{\mu}^{X_t}(\varphi_2)}(t)\right)^2}.
    \end{equation*}
    Using Proposition \ref{prop-sq-dist-convex} we get
    \begin{align*}
        L(\varphi_t)
        &\leq \expec{
            (1-t) d(X_{t+1}, \gamma_{\mu}^{X_t}(\varphi_1))^2
            + t d(X_{t+1}, \gamma_{\mu}^{X_t}(\varphi_2))^2}\\
        &\qquad - \expec{
            t(1-t) d(\gamma_{\mu}^{X_t}(\varphi_1), \gamma_{\mu}^{X_t}(\varphi_2))^2
        }\\
        &= (1-t)L(\varphi_1) + t L(\varphi_2) - t(1-t)\abs{\varphi_1 - \varphi_2}^2\expec{d(X_{t+1}, \mu)^2}.
    \end{align*}
    Since $X_{t+1}$ is $L^2(\Omega)$, we have that $\expec{d(X_{t+1}, \mu)^2} < \infty$, showing that $L$ is strongly convex.
  \end{proof}

\end{document}